\newif\ifsupp\supptrue
\newif\ifaistats\aistatsfalse

\ifaistats
\documentclass[twoside]{article}
\usepackage[accepted]{aistats2021}
\else
\documentclass{article}
\usepackage{fullpage}
\fi

\usepackage{amsmath,amssymb,pifont}
\usepackage{algorithm,algorithmic}
\usepackage{amstext,amssymb,amsmath}
\usepackage{amsthm}
\usepackage{comment}
\usepackage{multirow}
\usepackage{booktabs}
\usepackage[skip=0pt]{subcaption}
\usepackage{times}
\usepackage{lipsum}
\usepackage[shortlabels]{enumitem}
\usepackage{wrapfig}
\usepackage{array}
\usepackage{siunitx}
\usepackage{csvsimple}
\usepackage[multidot]{grffile}
\usepackage{bbm}
\usepackage{dblfloatfix}
\usepackage{makecell}
\usepackage{bbm, dsfont}
\usepackage{hyperref}
\usepackage{bbold}
\usepackage{mathtools}
\usepackage[dvipsnames]{xcolor}
\usepackage{natbib}

\renewcommand{\algorithmicrequire}{\quad\enspace\;\textbf{Input:}}
\renewcommand{\algorithmicensure}{\quad\enspace\;\textbf{Effects:}}
\renewcommand{\epsilon}{\varepsilon}
\newcommand{\mypar}[1]{\smallskip
	\noindent{\textbf{{#1}:}}}

\newcommand{\eps}{\ensuremath{\varepsilon}}

\newcommand{\boldb}{\ensuremath{\boldsymbol{b}}}

\newcommand{\boldg}{\ensuremath{\boldsymbol{g}}}

\newcommand{\boldu}{\ensuremath{\boldsymbol{u}}}
\newcommand{\boldv}{\ensuremath{\boldsymbol{v}}}

\newcommand{\boldx}{\bfx}

\newcommand{\boldR}{\bfR}

\newcommand{\bfR}{\ensuremath{\mathbf{R}}}

\newcommand{\bfx}{\ensuremath{\mathbf{x}}}

\newcommand{\bfz}{\ensuremath{\mathbf{z}}}

\newcommand{\calA}{\ensuremath{\mathcal{A}}}

\newcommand{\calC}{\ensuremath{\mathcal{C}}}
\newcommand{\calD}{\ensuremath{\mathcal{D}}}

\newcommand{\calL}{\ensuremath{\mathcal{L}}}
\newcommand{\calLc}{\ensuremath{\mathcal{L}^{(B)}_{\sf clipped}}}

\newcommand{\calN}{\ensuremath{\mathcal{N}}}

\newcommand{\calS}{\ensuremath{\mathcal{S}}}
\newcommand{\calT}{\ensuremath{\mathcal{T}}}

\newcommand{\calX}{\ensuremath{\mathcal{X}}}

\newcommand{\polylog}[1]{\ensuremath{{\rm polylog}\left(#1\right)}}
\renewcommand{\Pr}{\mathop{\mathbf{Pr}}}

\newcommand{\E}{\mathop{\mathbb{E}}}

\newtheorem{lem}{Lemma}[section]
\newtheorem{thm}[lem]{Theorem}
\newtheorem{cor}[lem]{Corollary}

\newtheorem{defn}[lem]{Definition}

\newtheorem{claim}[lem]{Claim}

\DeclareMathOperator*{\argmin}{arg\,min}
\newcommand{\bracket}[1]{\left(#1\right)}

\makeatletter
\newcommand{\vast}{\bBigg@{4}}
\newcommand{\Vast}{\bBigg@{5}}
\makeatother

\newcommand{\expt}[2]{\E_{#1}\left[{#2}\right]}
\newcommand{\power}[2]{\left(#1\right)^{#2}}
\newcommand{\indi}[1]{\mathds{1}\left(#1\right)}
\newcommand{\ip}[2]{\langle #1, #2\rangle}
\newcommand{\ipn}[3]{\langle #1, #2\rangle_{#3}}

\newcommand{\calLH}{\calLc}
\newcommand{\privT}{\theta^\texttt{priv}\,}
\newcommand{\thetaH}{\theta^*_{\sf {clipped}}\,}
\newcommand{\ltwo}[1]{\left\|#1\right\|_2}
\newcommand{\snorm}[2]{\left\|#1\right\|_{#2}}

\newcommand{\trunc}[2]{\left[#2\right]_{#1}}
\DeclarePairedDelimiter\abs{\lvert}{\rvert}
\newcommand{\bbracket}[1]{\left\{#1\right\} }

\newcommand{\bigO}[1]{O\left(#1\right)}
\newcommand{\bigOmega}[1]{\Omega\left(#1\right)}
\newcommand{\expo}[1]{\exp\bracket{#1}}
\newcommand{\bT}{\Theta}%
\newcommand{\bt}{\theta}%
\newcommand{\ee}[1]{E^{(#1)}}%

\newcommand{\interior}[1]{{\kern0pt#1}^{\mathrm{o}}}

\newcommand{\adpgd}{\calA_{\sf DP\textsf{-}GD}}
\newcommand{\adpgdc}{\calA^{\sf clipped}_{\sf DP\textsf{-}GD}}
\newcommand{\gpriv}{\boldg^{\texttt{priv}}}

\newcommand{\dpgdclip}{DP-GD$_{\sf clipped}$}
\newcommand{\rank}{\texttt{rank}}
\newcommand{\clip}{\texttt{{clip}}}
\newcommand{\ellc}{\ell^{(B)}_{\sf clipped}}
\newcommand{\subg}{\partial_\theta}
\renewcommand{\tilde}{\widetilde}

\newcommand{\cn}{B} %
\newcommand{\thetapop}{\theta_{\sf pop}^*}
\newcommand{\opt}{\alpha_{\sf opt}}
\newcommand{\gen}{\alpha_{\sf gen}}

\newcommand*\samethanks[1][\value{footnote}]{\footnotemark[#1]}

\begin{document}
\ifaistats
\twocolumn[
\aistatstitle{Evading the Curse of Dimensionality in Unconstrained Private GLMs}

\aistatsauthor{ Shuang Song \And Thomas Steinke \And  Om Thakkar \And Abhradeep Thakurta }

\aistatsaddress{Google \And Google \And Google \And Google} ]

\else
\title{Evading Curse of Dimensionality in Unconstrained Private GLMs\\ via Private Gradient Descent}
\author{Shuang Song\thanks{Google. \texttt{\{shuangsong, steinke, omthkkr, athakurta\}@google.com}} \and Thomas Steinke\samethanks[1] \and Om Thakkar\samethanks[1] \and Abhradeep Thakurta\samethanks[1]}

\maketitle
\fi

\begin{abstract}
We revisit the well-studied problem of differentially private empirical risk minimization (ERM). We show that for unconstrained convex generalized linear models (GLMs), one can obtain an excess empirical risk of $\tilde O\left(\sqrt{\rank}/\epsilon n\right)$, where $\rank$ is the rank of the feature matrix in the GLM problem, $n$ is the number of data samples, and $\epsilon$ is the privacy parameter. This bound is attained via differentially private gradient descent (DP-GD). Furthermore, via the \emph{first lower bound for unconstrained private ERM}, we show that our upper bound is tight. In sharp contrast to the constrained ERM setting, there is no dependence on the dimensionality of the ambient model space ($p$). (Notice that $\rank\leq \min\{n, p\}$.) 
Besides, we obtain an analogous excess population risk bound which depends on $\rank$ instead of $p$.

For the smooth non-convex GLM setting (i.e., where the objective function is non-convex but preserves the GLM structure), we  further show that DP-GD attains a  dimension-independent convergence of $\tilde O\left(\sqrt{\rank}/\epsilon n\right)$ to a first-order-stationary-point of the underlying objective.

Finally, we show that for convex GLMs, a variant of DP-GD commonly used in practice (which involves clipping the individual gradients) also exhibits the same dimension-independent convergence to the minimum of a well-defined objective. To that end, we provide a structural lemma that characterizes the effect of clipping on the optimization profile of DP-GD.
\end{abstract}

\section{Introduction}
\label{sec:intro}

Differentially private empirical risk minimization (ERM) is a well-studied area in the privacy literature~\citep{chaudhuri2011differentially,kifer2012private,song2013stochastic,BST14,jain2014near,DP-DL,mcmahan2017learning, WLKCJN17, BassilyFTT19,iyengar2019towards, pichapati2019adaclip,TAB19,feldman2019private}. In the constrained setting, where the model space is bounded by $\calC\subsetneq\bfR^p$, tight upper and lower bounds are known for both excess empirical risk~\citep{BST14}, and excess population risk~\citep{bassily2019private,feldman2019private}. Surprisingly, in the arguably simpler unconstrained setting where $\calC=\bfR^p$, the problem space is much less explored. To our knowledge, the only prior work that distinguishes between the constrained and the unconstrained case is that of~\citet{jain2014near}. They show dimension-independent upper bounds for population risk (under differential privacy) in the convex generalized linear models (GLMs) case, which alludes to a separation between  constrained and unconstrained settings. In contrast, the lower bound of \citet{BST14} shows that an explicit dependence on the dimensionality ($p$) is necessary in the constrained setting, even for GLMs.

In this work, we revisit the private unconstrained ERM setting, and close the gap between upper and lower bounds for the GLM case. We first show that for convex GLMs, one can attain an excess empirical risk of $\tilde O\left(\sqrt{\rank}/(\epsilon\cdot n)\right)$\footnote{$\tilde O(\cdot)$ hides $\polylog{1/\delta}$, where $\delta$ is a privacy parameter.}, where $\rank$ is the rank of the feature matrix, and $\epsilon$ is the privacy parameter. 
In comparison, \citet{jain2014near} provide a much worse upper bound on the excess population risk for the unconstrained setting as $\tilde O\left(1/\epsilon\sqrt{n}\right)$. If interpreted in terms of excess empirical risk, their bound matches ours only when $\rank=n$. 
Our upper bound is in sharp contrast to the lower bound of $\tilde \Omega\left(\sqrt{p}/(\epsilon\cdot n)\right)$ for the constrained setting~\citep{BST14} ($p$ being the dimensionality of the model space), as $\rank\leq\min\{n,p\}$ always holds, and $\rank$ may be much smaller. 
Our upper bound is achieved via differentially private gradient descent (DP-GD)~\citep{BST14,song2013stochastic,talwar2014private,DP-DL}. While our guarantees extend to the stochastic variant of DP-GD, we focus on the full gradient version for brevity.

We further show that our bound on the excess empirical risk is essentially tight. It is worth mentioning that ours is \emph{the first lower bound} on excess empirical risk for any unconstrained ERM problem (including GLMs). The lower bound is based on fingerprinting codes~\citep{bun2018fingerprinting,SteinkeU17}.

In subsequent works~\citep{kairouz2020dimension,zhou2020bypassing}, there have been extensions of our results for convex GLMs to general convex ERMs via adaptive preconditioners. However, these results are more restrictive in their guarantees, and \emph{do not imply our results}. 
They require the existence of public data to identify the subspace where the gradient of the objective function lies.

Going beyond convex GLMs, we show that the dimension-independent convergence holds even for \emph{non-convex} GLMs (i.e., the loss function has the GLM structure, but can be non-convex). Such problems appear commonly in robust regression \citep{amid2019robust,masnadi2009design, masnadi2010design}. 
Given an objective function $\calL(\theta;D)=\frac{1}{n}\sum\limits_{i=1}^n\ell(\theta;d_i)$, where dataset $D = \{d_1, \ldots, d_n\}$, we show that DP-GD reaches a first-order stationary point (FOSP) of $\calL(\theta;D)$ at a rate of $\tilde O\left(\sqrt{\rank} /(\epsilon\cdot n)\right)$ as long as the individual loss functions $\ell$ are smooth in the model parameter. 

All our upper bounds are primarily based on DP-GD, which in general, requires a bound on the $\ell_2$-norm of the subgradients of individual loss functions in $\calL(\theta;D)$. In practice, however, such a bound is seldom known a priori for complex models. As a result, a variant of DP-GD, called \emph{clipped} DP-GD~\citep{DP-DL,papernot2020tempered}, is commonly used. It scales the subgradients down if the $\ell_2$-norm crosses a predefined threshold, a.k.a.~the clipping norm ($\cn>0$). In this work, we show that no matter what the clipping norm is, for convex GLMs, clipped DP-GD still has an excess empirical risk of $\tilde O(\sqrt{\rank}/(\epsilon\cdot n))$ with respect to a well-defined objective function $\calLc(\theta;D)$ (in contrast to the original objective $\calL(\theta;D)$). The function $\calLc$ still satisfies the convex GLM property. While there are other contemporary works~\citep{chen2020understanding} that study the effect of clipping on DP-GD, they are orthogonal to the results in this paper. We focus on formal excess empirical risk guarantees, whereas \citet{chen2020understanding} focus on understanding the gradient profile generated by DP-GD due to clipping. 

\subsection{Our Contributions}
\label{sec:contrib}

\mypar{Dimension-independent excess empirical risk bounds for convex generalized linear models (GLMs)} In Section~\ref{sec:ubdpgd}, we consider a class of problems with loss functions of the form $\ell(\ip{\bt}{\bfx};y)$, where $\bfx\in\calX\subset\bfR^p$ is the feature vector, $y\in\bfR$ is the response variable, and $\ell$ is convex in the first parameter. We show that if the optimization is over an unconstrained space (i.e., $\theta\in\bfR^p$), then for an objective function $\calL(\theta;D)=\frac{1}{n}\sum\limits_{i=1}^n\ell\left(\ip{\theta}{\bfx_i};y_i\right)$, one can achieve an excess empirical risk  of $\widetilde O\left(L\sqrt{\rank}/(\epsilon\cdot n)\right)$, where ${\rank}\leq \min\{n,p\}$ is the rank of the feature matrix $X=[\bfx_1,\ldots,\bfx_n]$. 
To the best of our knowledge, \emph{this is the first rank-based excess empirical risk bound for private convex GLMs}. Notice that the bound does not have any explicit dependence on the dimensionality $p$. We achieve this bound by optimizing on $\calL(\theta;D)$ using \emph{differentially private gradient descent} (DP-GD)~\citep{BST14,talwar2014private,song2013stochastic}. We also obtain an excess population risk of the form $\widetilde{O}(L\cdot\min\{1/\sqrt n,\sqrt{\rank}/(\epsilon\cdot n)\})$, which is equivalent to the optimal excess population risk obtained by prior works~\citep{bassily2019private,bassily2020stability,BassilyFTT19}, except that the ambient dimensionality $p$ is replaced by $\rank$\footnote{In the context of population risk, $\rank$ refers to the rank of the covariance matrix for the data generating distribution.}.

Existing lower bounds for constrained private convex learning~\citep{BST14} (i.e., $\theta\in\calC\subsetneq\bfR^p$) show that for excess empirical risk, an explicit polynomial dependence on the dimensionality of the model space ($p$) is necessary.
In contrast, our bound only depends on the  $\rank$ of the feature matrix $X$.
Our main insight is that for DP-GD on generalized linear problems, the gradients lie in a low-rank subspace. 
The noisy gradients that DP-GD uses for state updates do not significantly impact this low-rank structure due to the \emph{spherical} (and \emph{stable}) nature of the Gaussian distribution. Our  results  extend to the local differentially private (LDP)~\citep{Warner,evfimievski2003limiting,KLNRS} variant of DP-GD, albeit with an increase of a $\sqrt{n}$ factor in the excess empirical risk~\citep{duchi2018minimax}.

While \citet{jain2014near} proved a related dimension-independent risk guarantee for two other differentially private algorithms, namely output perturbation \citep{chaudhuri2011differentially} and objective perturbation \citep{chaudhuri2011differentially,kifer2012private}, our result is notable in the following aspects. First, we provide a more fine-grained control via the ${\rank}$ parameter. The result in \citet{jain2014near} only provides guarantees where ${\rank}$ is upper-bounded by $n$. Second, \citet{jain2014near} crucially relies on the existence of a centralized data source, whereas our result extends seamlessly to the LDP setting. Third, unlike the algorithms in \citet{jain2014near}, DP-GD does not require convexity to ensure privacy. This is important because even if the overall optimization function is non-convex, DP-GD still ensures differential privacy \citep{BST14,DP-DL}. Depending on the optimization profile, we may still observe a dimension-independent convergence. We provide more evidence of this phenomenon in Section~\ref{sec:privNonConvex}.

Additionally, we obtain a population risk guarantee that is asymptotically the same as the optimal excess population risk in~\cite{BassilyFTT19}, except that the ambient dimensionality ($p$) is replaced by $\rank$.
In Section~\ref{sec:expts}, we provide empirical evidence demonstrating the dimension-independence of DP-GD with Gaussian noise on logistic regression. 

\mypar{Tight lower bound on excess empirical risk for convex GLMs} In Section~\ref{sec:lbfp}, we show that our dimension-independent upper bound on the excess empirical risk for unconstrained convex GLM achieved via DP-GD is tight. This lower bound is in sharp contrast to that in \citet{BST14}, where they show that for GLMs in the constrained setting, an explicit dependence on the dimensionality ($p$) is necessary. It is worth mentioning that \emph{our lower bound is the first for any unconstrained private convex ERM}. Our lower bound is proved by transforming a GLM instance (namely, $\ell(\ip{\theta}{\bfx};y)=|\ip{\bfx}{\theta}-y|$) to estimating one-way marginals, to which we can apply fingerprinting techniques~\citep{bun2018fingerprinting,SteinkeU17}.

\mypar{Dimension-independent convergence to a first-order stationary point for non-convex GLMs} In Section \ref{sec:privNonConvex}, we extend our dimension-independent result to non-convex generalized linear problems, i.e., where the loss function $\ell$ can be non-convex but preserves the inner-product structure. We show that for this class of problems, DP-GD converges to a first-order stationary point (FOSP) (i.e., where the gradient of the objective function is zero). Again, this convergence guarantee is independent of the model dimensionality, and only depends on ${\rank}$ of the feature matrix. 
Specifically, we show that if the loss function for the non-convex generalized linear problem is smooth and $L$-Lipschitz in the $\ell_2$ norm, then DP-GD (paired with the exponential mechanism~\citep{mcsherry2007mechanism}) outputs a model $\theta_{\sf priv}$ such that the gradient of the objective function $\calL(\theta;D)$ at $\theta_{\sf priv}$ has $\ell_2$-norm of  $\widetilde O\left(L\sqrt{\rank}/(\epsilon n)\right)$.

While there has been work on understanding the convergence of variants of DP-GD on non-convex losses \citep{wang2019differentially}, ours is the first result to demonstrate a dimension-independent convergence. At the heart of our result is a simple folklore argument stated in \citet{allen2018natasha} that shows first-order convergence of GD for non-convex objectives. We conjecture that our result can be extended to second-order convergence (analogous to \citet{wang2019differentially}) under additional assumptions on the loss function. A natural direction would be to modify the argument of \citet{jin2017escape} to make it amenable to DP-GD.

\mypar{Analysis of clipped differentially private gradient descent (DP-GD) on convex GLMs} While the upper bounds in this paper are achieved by DP-GD, one major caveat for using the algorithm in practice is that it requires a predefined upper bound of $L$ on $\ltwo{\partial_\theta\ell(\ip{\theta}{\bfx};y)}$ for all $\bfx,y$. (Here, $\partial_\theta$ corresponds to the subgradients.) 
In real-world applications, $L$ is almost never known a priori. As a result, a variant of DP-GD (called clipped DP-GD) is used in practice where the individual subgradients corresponding to each data sample $(\bfx,y)$ are scaled/clipped to ensure that they are upper bounded by a predefined quantity $B>0$, a.k.a. the clipping norm~\citep{DP-DL,papernot2020tempered,chen2020understanding}. In Section~\ref{sec:clipping},  we show that no matter what the clipping norm $B$ is, for convex GLMs, clipped DP-GD optimizes a well-defined convex objective (denoted by $\calLc(\theta;D)$) corresponding to $\calL(\theta;D)$. Furthermore, the excess empirical risk with respect to $\calLc$ is $\widetilde O\left(B\sqrt{\rank}/(\epsilon\cdot n)\right)$, with $\rank$ being the rank of the feature matrix $X$. Notice the same dimension-independent convergence for clipped DP-GD as that of vanilla DP-GD. We also show that if $B\geq L$, then $\calLc(\theta;D)$ equals $\calL(\theta;D)$ point-wise. To prove the above bound, we provide a structural lemma (see Section~\ref{sec:clipeqvhuber}, which characterizes the clipping operation as a variant of Huberization~\citep{huber} for convex GLMs. To our knowledge, this is the \emph{first convergence guarantee for clipped DP-GD}.

As an interlude, in Section~\ref{sec:negclip}, we show that the convergence guarantees of clipped DP-GD are sensitive to the choice of the clipping norm $B$. Setting it low (i.e., $B\ll L$) can result in strange behaviors in the optimization profile, ranging from introducing $\Omega(1)$ bias in the excess empirical risk, to generating vectors that do not conform to the gradient field of any ``natural'' convex function. 

We note that there is a line of work on the practice and theory of gradient clipping~\citep{goodfellow2016deep,pascanu2012understanding,pascanu2013difficulty,zhang2019gradient}. Despite the similarity in name, these algorithms are different as they clip the \emph{averaged} gradient in each step, while in clipped DP-GD, we need the \emph{individual} gradient to be clipped to get a reasonable privacy/utility trade-off.

\section{Preliminaries}
\label{sec:diff}

In this section, we provide the some of the concepts required in the rest of the paper. %

\begin{defn}[Seminorm]
Given a vector space $V$ over a field $F$ of the real numbers $\boldR$, a seminorm on $V$ is a nonnegative-valued function $\rho: V\to\boldR$ with the following properties. For all $a\in F$, and $\boldu,\boldv\in V$: 
\begin{enumerate}
    \item {\bf Triangle inequality}: $\rho(\boldu+\boldv)\leq \rho(\boldu)+\rho(\boldv)$.\vspace{-0.2cm}
    \item {\bf Absolute scalability}: $\rho(a\cdot\boldu)=|a|\cdot\rho(\boldu)$.
\end{enumerate}
\label{def:seminorm}
\end{defn}

\mypar{Lipschitzness, Convexity, and Smoothness}
We additionally require the following definitions to state our results. These properties usually govern the rate of convergence of an algorithm for optimizing ERMs. 

\begin{defn}[$\ell_2$-Lipschitz continuity]
A function $f:\calC\to\boldR$ is $L$-Lipschitz w.r.t. the $\ell_2$-norm over a set $\calC\subseteq\boldR^p$ if the following holds:
$\forall \theta_1,\theta_2\in\calC, \left|f(\theta_1)-f(\theta_2)\right|\leq L \cdot\ltwo{\theta_1-\theta_2}$.
\label{def:lip}
\end{defn}

\begin{defn}[(Strong) convexity w.r.t. $\ell_2$-norm]
A function $f:\calC\to\boldR$ is $\Delta$-strongly convex w.r.t. the $\ell_2$-norm over a set $\calC\subseteq\boldR^p$ if $\forall \alpha\in(0,1),(\theta_1,\theta_2)\in \calC\times\calC$: 
\ifaistats
{\small\begin{align*}
& f(\alpha\theta_1+(1-\alpha)\theta_2) \\ & \qquad \leq
\alpha f(\theta_1)+(1-\alpha)f(\theta_2)-\Delta\frac{\alpha(\alpha-1)}{2}\ltwo{\theta_1-\theta_2}^2. 
\end{align*}}
\else
\begin{align*}
& f(\alpha\theta_1+(1-\alpha)\theta_2) 
\leq \alpha f(\theta_1)+(1-\alpha)f(\theta_2)-\Delta\frac{\alpha(\alpha-1)}{2}\ltwo{\theta_1-\theta_2}^2. 
\end{align*}
\fi
\noindent Function $f$ is simply convex if the above holds for $\Delta=0$.
\label{def:sc}
\end{defn}

\begin{defn}[Smoothness]
A function $f:\calC\to\boldR$ is $\beta$-smooth on $\calC\subseteq\boldR^p$ if for all $\theta_1\in \calC$ and for all $\theta_2\in\calC$, we have $
f(\theta_2) \leq f(\theta_1) + \ip{\nabla f(\theta_1)}{\theta_2-\theta_1} + \frac{\beta}{2} \|\theta_1-\theta_2\|_2^2.
$
\label{def:smooth}
\end{defn}

\mypar{Differential Privacy} In this paper, we focus on approximate differential privacy (DP) \citep{DMNS, ODO}. 

\begin{defn}[Differential privacy \citep{DMNS, ODO}] A randomized  algorithm $\calA$ is $(\eps,\delta)$-differentially private if, for any pair of datasets $D$ and $D'$ differing in exactly one data point (i.e., one data point is replaced in the other), and for all events $\calS$ in the output range of $\calA$, we have 
$$\Pr[\calA(D)\in \calS] \leq e^{\eps} \cdot \Pr[\calA(D')\in \calS] +\delta,$$
where the probability is taken over the random coins of $\calA$. 
\label{def:diiffP}
\end{defn}
For meaningful privacy guarantees, $\epsilon$ is assumed to be a small constant,  and $\delta \ll 1/n$ for $n=|D|$.

\mypar{Empirical risk minimization (ERM)} Let $D=\{d_1,\cdots,d_n\}\subseteq\calD^n$ be a data set of $n$ samples drawn from the domain $\calD$, and for $\calC\subseteq\boldR^p$ being the model space, let $\ell:\calC\times\calD\to\boldR$ be a loss function. Then the empirical risk over the data set $D$ is defined as $\calL(\theta;D)=\frac{1}{n}\sum\limits_{i=1}^n\ell(\theta;d_i)$. The objective of an \emph{empirical risk minimization} (ERM) algorithm is to output a model $\theta\in\calC$ that approximately minimizes the empirical risk $\calL$ over the set $\calC$. For the theoretical guarantees in this paper, we will only look at ERM loss, and the \emph{excess empirical risk} 
$R(\theta) = \calL(\theta;D) - \calL(\theta^*;D)$, where $\theta^* = \argmin_{\theta\in\calC} \calL(\theta;D)$.
By stability-based arguments \citep{BST14,SSSS09}, one can easily translate excess empirical risk for differentially private algorithms to their corresponding \emph{excess population risk}, where the population risk for $\theta$ is defined as $\mathbb{E}_{d\sim\calT}\left[\ell(\theta;d)\right]$, with $\calT$ being a given distribution over $\calD$. 

If $\calC \subsetneq \boldR^p$, we are in the so-called constrained ERM setting, whereas $\calC = \boldR^p$ is denoted as the unconstrained setting. In this paper, we focus on the \emph{unconstrained} setting.

\mypar{Generalized Linear Models} For most of this paper, we focus on a special class of ERM problems called generalized linear models~\citep{SSSS09}, where the loss function $\ell(\theta;d)$ takes a special inner-product form $\ell(\ip{\theta}{\boldx};y)$ for $d = (\boldx, y)$.
Here, $\boldx\in\boldR^p$ is usually called the feature vector and $y\in\boldR$ the response. 
Instead of being the feature vector in the original data, $\boldx$ can also be thought of as representing a mapped value $\phi(\boldx)$ of original feature vector. We do not make the distinction here. 

\mypar{Differentially Private Gradient Descent}
\label{sec:dpSGD}
We provide a formal description of Differentially Private Gradient Descent (DP-GD) in Algorithm~\ref{Alg:Frag}. 
In this version, the gradient $\boldg_t$ is computed over the whole data set, and the output $\privT$ is the average of the models over all iterations. 
In practice, we may instead use differentially private stochastic gradient descent (DP-SGD), where $\boldg_t$ is computed over a random mini-batch, and the output $\privT$ is the last model. 
While our analytical results are for the former setting (due to brevity), they extend to the latter with mild modifications to the proofs.
\begin{algorithm}[t]
	\caption{$\adpgd$: Diff. private gradient descent}
	\begin{algorithmic}[1]
	    \REQUIRE Data set $D=\{d_1,\cdots,d_n\}$, loss function: $\ell:\boldR^p\times\calD\to\boldR$, gradient $\ell_2$-norm bound: $L$, constraint set: $\calC\subseteq\boldR^p$, number of iterations: $T$, noise variance: $\sigma^2$, learning rate: $\eta$.
	    \STATE $\theta_0\leftarrow \boldsymbol{0}$.
	    \FOR{$t = 0,\dots,T-1$}
	        {\STATE $\gpriv_t \leftarrow \frac{1}{n}\sum\limits_{i=1}^n\subg \ell(\theta_t;d_i)+\calN\left(0,\sigma^2\right)$. \label{alg: stepgrad}}
	        {\STATE { $\theta_{t+1}\leftarrow \Pi_\calC\left(\theta_t-\eta\cdot\gpriv_t\right)$}, where\\  {$\Pi_\calC(\boldv)=\argmin\limits_{\theta\in\calC}\|\boldv-\theta\|_2$}.}
	    \ENDFOR
	    \STATE {\bf return} $\privT=\frac{1}{T}\sum\limits_{t=1}^T\theta_t$.
	\end{algorithmic}
	\label{Alg:Frag}
\end{algorithm}

\begin{thm}[From \citet{DP-DL,mironov2017renyi}]
If $\|\subg \ell(\theta;d)\|_2 \leq L$ for any $\theta \in \calC$ and $d \in \boldR^p$, then, Algorithm $\adpgd$ (Algorithm \ref{Alg:Frag}) is $(\epsilon,\delta)$-differentially private if the noise variance is $\sigma^2=\frac{2L^2T\log(1/\delta)}{(n\epsilon)^2}$.
\label{thm:priv}
\end{thm} 

\begin{thm}[From \citet{BST14,talwar2014private}]
If the constraint set $\calC$ is convex, the loss function $\ell(\theta;d)$ is convex in the first parameter, $\|\subg\ell(\theta;d)\|_2\leq L$ for all $\theta\in\calC$ and $d\in D$,
then for objective function $\calL\left(\theta;D\right)=\frac{1}{n}\sum\limits_{i=1}^n\ell(\theta;d_i)$, under appropriate choices of the learning rate and the number of iterations in Algorithm $\adpgd$ (Algorithm \ref{Alg:Frag}), we have
\ifaistats
\begin{align*}
&\mathbb{E}\left[\calL\left(\privT;D\right)\right]-
\calL\left (\theta^*;D\right) 
\\
& \qquad \leq\frac{L\|\theta_0 - \theta^*\|_2\sqrt{p\log(1/\delta)}}{\epsilon n},
\end{align*}
\else
\begin{align*}
\mathbb{E}\left[\calL\left(\privT;D\right)\right]-\calL\left (\theta^*;D\right) 
\leq\frac{L\|\theta_0 - \theta^*\|_2\sqrt{p\log(1/\delta)}}{\epsilon n},
\end{align*}
\fi
where $\theta^* = \argmin\limits_{\theta\in\calC} \calL(\theta;D)$ is the minimizer and $\theta_0 \in \calC$ is the initial model.\\
The corresponding high-probability version is as follows. With probability at least $1-\beta$, we have
\ifaistats
\begin{align*}
&\calL\left(\privT;D\right) - \calL\left (\theta^*;D\right) \\
& \qquad
\leq\frac{L\|\theta_0-\theta^*\|_2\sqrt{p\log(1/\delta)\log(1/\beta)}}{\epsilon n}.
\end{align*}
\else
\begin{align*}
&\calL\left(\privT;D\right) - \calL\left (\theta^*;D\right)
\leq\frac{L\|\theta_0-\theta^*\|_2\sqrt{p\log(1/\delta)\log(1/\beta)}}{\epsilon n}.
\end{align*}
\fi
\label{thm:dpsgd_convergence}
\end{thm}
\section{Tight Excess Empirical Risk for Convex Generalized Liner Models (GLMs)}
\label{sec:ucconvex}

In Section~\ref{sec:ubdpgd}, we show that algorithm $\adpgd$ (Algorithm~\ref{Alg:Frag}) converges to an excess empirical risk of $\tilde O(\sqrt{\rank}/{\epsilon n})$, which in particular is independent of the ambient dimensionality ($p$). (Here $\rank$ is the rank of the feature matrix for GLMs.) In comparison, prior excess empirical risk of $\tilde O(\sqrt{p}/{\epsilon n})$ (via Theorem
~\ref{thm:dpsgd_convergence}) has an explicit dependence $p$ for Algorithm $\adpgd$.

In Section~\ref{sec:lbdpgd} we further show that this bound is tight. It is the \emph{first lower bound on excess empirical risk for any unconstrained Lipschitz optimization problem}.

\subsection{Upper-bound via Private Gradient Descent}
\label{sec:ubdpgd}
Consider the following convex optimization problem. Let $\calL(\theta;D)=\frac{1}{n}\sum\limits_{i=1}^n\ell(\ip{\theta}{\bfx_i};y_i)$ be an objective function defined over the data set $D=\{(\bfx_1,y_1),\ldots,(\bfx_n,y_n)\}$ with $\bfx_i\in\calX$ and $y_i\in\bfR$ for all $i\in[n]$, where $\calX\subset \bfR^p$ is a bounded set. 
Assume the loss function $\ell(\ip{\theta}{\bfx};y)$ is convex in its first parameter and is $L$-Lipschitz with respect to the $\ell_2$-norm over all $\theta\in\bfR^p$ and for all $\bfx$ and $y$. 
In other words, the $\ell_2$-norm of $\partial_{\theta}\ell$ is upper bounded by $L$.
The objective is to output $\privT$ that approximately solves $\argmin\limits_{\theta\in\bfR^p}\calL(\theta;D)$ while satisfying differential privacy. In the following, we show that the excess empirical risk for $\adpgd$ (Algorithm~\ref{Alg:Frag}) scales approximately as $\widetilde O(\sqrt{\rank}/(\epsilon n))$. Here $\rank$ is the rank of the feature matrix $[\bfx_1 \vert \bfx_2 \vert \cdots \vert \bfx_n]$ formed by stacking the feature vectors as columns.
In Section~\ref{sec:lbdpgd}, we show that this bound is indeed tight.

\begin{thm}
Let $\theta_0=\mathbf{0}^p$ be the initial point of $\adpgd$. Let $\theta^*=\argmin\limits_{\theta\in\bfR^p}\calL(\theta;D)$, and $M$ be the projector to the eigenspace of the matrix $\sum\limits_{i=1}^n \bfx_i\bfx_i^T$. 
Letting $L$ be the gradient $\ell_2$-norm bound. Setting the constraint set $\calC=\boldR^p$ and running $\adpgd$ on $\calL(\theta;D)$ for $T=n^2\epsilon^2$ steps with appropriate learning rate $\eta$, we get
\ifaistats
\begin{align*}
&\mathbb{E}\left[\calL\left(\privT;D\right)\right]-\calL\left (\theta^*;D\right)\\&
\leq\frac{L\snorm{\theta^*}{M}\sqrt{1 + 2\cdot\rank(M)\cdot \log(1/\delta)}}{\epsilon n}.
\end{align*}
\else
\begin{align*}
&\mathbb{E}\left[\calL\left(\privT;D\right)\right]-\calL\left (\theta^*;D\right)
\leq\frac{L\snorm{\theta^*}{M}\sqrt{1 + 2\cdot\rank(M)\cdot \log(1/\delta)}}{\epsilon n}.
\end{align*}
\fi
Here, $\rank(M)\leq n$ (but can be much smaller), and $\snorm{\cdot}{M}$ is the seminorm w.r.t. the projector $M$.
\label{thm:dimIndependence}
\end{thm}

Though our result is for excess empirical risk, it can be translated to excess population risk guarantees via standard stability-based arguments \citep{bassily2020stability}. (We provide the details below.) The crux of our proof technique is to work in the subspace generated by the feature vectors for generalized linear problem. We proved the guarantees for DP-GD that uses full gradient and returns the average of the models over all iterations. Our proof would extend seamlessly (by modifying the proofs of Theorems 1 and 2 in \citet{shamir2013stochastic}) to settings where stochastic gradients over mini-batches are used and the final model is returned as the output.

\begin{proof}[Proof of Theorem~\ref{thm:dimIndependence}]
We prove the theorem via the standard template for analyzing SGD methods \citep{bubeck2015convex}. Recall $\privT=\frac{1}{T}\sum\limits_{t=1}^T\theta_t$, where $\{\theta_1,\ldots,\theta_T\}$ are the models in each iterate of DP-GD. 
Let $\boldg_t$ denote any subgradient in $\partial \calL(\theta_t;D)$.
By convexity and the standard linearization trick in convex optimization \citep{bubeck2015convex}, we have:
\begin{align}
\calL\left(\privT;D\right)-\calL\left(\theta^*;D\right)&\leq \frac{1}{T}\sum\limits_{t=1}^T\ip{\boldg_t}{\theta_t-\theta^*}
\label{eq:12as}
\end{align}
Let $V$ be the eigenbasis of $\sum\limits_{i=1}^n \bfx_i \bfx_i^T$ and let $M=VV^T$. $M$ is a positive semidefinite matrix and it defines a seminorm $\snorm{\cdot}{M}$ (by Definition \ref{def:seminorm}).
Let $\boldb_t$ be the Gaussian noise vector added at time step $t$.
To bound the error in \eqref{eq:12as}, we will use a potential argument w.r.t. the potential function 
\ifaistats
\begin{align*}
\Psi_{t}(\theta)&=\mathbb{E}_{\boldb_1,\ldots,\boldb_t}\left[\snorm{\theta-\theta^*}{M}^2\right]\\
&=\mathbb{E}_{\boldb_1,\ldots,\boldb_{t-1}}\left[\mathbb{E}_{\boldb_t}\left[\left.\snorm{\theta-\theta^*}{M}^2\right\vert \boldb_1,\ldots,\boldb_{t-1}\right]\right].
\end{align*}
\else
\begin{align*}
\Psi_{t}(\theta)&=\mathbb{E}_{\boldb_1,\ldots,\boldb_t}\left[\snorm{\theta-\theta^*}{M}^2\right]=\mathbb{E}_{\boldb_1,\ldots,\boldb_{t-1}}\left[\mathbb{E}_{\boldb_t}\left[\left.\snorm{\theta-\theta^*}{M}^2\right\vert \boldb_1,\ldots,\boldb_{t-1}\right]\right].
\end{align*}
\fi
Recall that the update step in DP-GD is $\theta_{t+1}\leftarrow\theta_t-\eta\left(\boldg_t+\boldb_t\right)$. We get the following by simple algebraic manipulation:
\ifaistats
{\begin{align}
\Psi_{t}(\theta_{t+1})&=\mathbb{E}_{\boldb_1,\dots,\boldb_t}\left[\snorm{\left(\theta_{t}-\theta^*\right)-\eta(\boldg_t+\boldb_t)}{M}^2\right]\nonumber\\
&=\Psi_{t}(\theta_t)-2\eta\mathbb{E}_{\boldb_1,\dots,\boldb_t}\left[\ipn{\boldg_t+\boldb_t}{\theta_t-\theta^*}{M}\right]\nonumber\\
&+\eta^2\mathbb{E}_{\boldb_1,\dots,\boldb_t}\left[\snorm{\boldg_t+\boldb_t}{M}^2\right]
\end{align}
\begin{align}
&=\Psi_{t}(\theta_t)-2\eta\mathbb{E}_{\boldb_1,\dots,\boldb_t}\left[\ip{\boldg_t+\boldb_t}{\theta_t-\theta^*}\right]\nonumber\\
&+\eta^2\mathbb{E}_{\boldb_1,\dots,\boldb_t}\left[\snorm{\boldg_t+\boldb_t}{M}^2\right]\label{eq:abc132}\\
&\leq\Psi_{t}(\theta_t)-2\eta\mathbb{E}_{\boldb_1,\dots,\boldb_{t-1}}\left[\ip{\boldg_t}{\theta_t-\theta^*}\right]\nonumber\\  &+\eta^2\left(L^2+\mathbb{E}_{\boldb_t}\left[\snorm{\boldb_t}{M}^2\right]\right)\nonumber\\
&=\Psi_{{t-1}}(\theta_t)-2\eta\mathbb{E}_{\boldb_1,\dots,\boldb_{t-1}}\left[\ip{\boldg_t}{\theta_t-\theta^*}\right]\nonumber\\  &+\eta^2\left(L^2+\mathbb{E}_{\boldb_t}\left[\snorm{\boldb_t}{M}^2\right]\right)\nonumber\\
&=\Psi_{{t-1}}(\theta_t)-2\eta\mathbb{E}_{\boldb_1,\dots,\boldb_{t-1}}\left[\ip{\boldg_t}{\theta_t-\theta^*}\right]\nonumber\\  &+\eta^2\left(L^2+\rank(M)\cdot \sigma^2\right).
\label{eq:anc143}
\end{align}}
\else
\begin{align}
\Psi_{t}(\theta_{t+1})&=\mathbb{E}_{\boldb_1,\dots,\boldb_t}\left[\snorm{\left(\theta_{t}-\theta^*\right)-\eta(\boldg_t+\boldb_t)}{M}^2\right]\nonumber\\
&=\Psi_{t}(\theta_t)-2\eta\mathbb{E}_{\boldb_1,\dots,\boldb_t}\left[\ipn{\boldg_t+\boldb_t}{\theta_t-\theta^*}{M}\right]+\eta^2\mathbb{E}_{\boldb_1,\dots,\boldb_t}\left[\snorm{\boldg_t+\boldb_t}{M}^2\right] \\
&=\Psi_{t}(\theta_t)-2\eta\mathbb{E}_{\boldb_1,\dots,\boldb_t}\left[\ip{\boldg_t+\boldb_t}{\theta_t-\theta^*}\right]+\eta^2\mathbb{E}_{\boldb_1,\dots,\boldb_t}\left[\snorm{\boldg_t+\boldb_t}{M}^2\right]\label{eq:abc132}\\
&\leq\Psi_{t}(\theta_t)-2\eta\mathbb{E}_{\boldb_1,\dots,\boldb_{t-1}}\left[\ip{\boldg_t}{\theta_t-\theta^*}\right]+\eta^2\left(L^2+\mathbb{E}_{\boldb_t}\left[\snorm{\boldb_t}{M}^2\right]\right)\nonumber\\
&=\Psi_{{t-1}}(\theta_t)-2\eta\mathbb{E}_{\boldb_1,\dots,\boldb_{t-1}}\left[\ip{\boldg_t}{\theta_t-\theta^*}\right]+\eta^2\left(L^2+\mathbb{E}_{\boldb_t}\left[\snorm{\boldb_t}{M}^2\right]\right)\nonumber\\
&=\Psi_{{t-1}}(\theta_t)-2\eta\mathbb{E}_{\boldb_1,\dots,\boldb_{t-1}}\left[\ip{\boldg_t}{\theta_t-\theta^*}\right]+\eta^2\left(L^2+\rank(M)\cdot \sigma^2\right).
\label{eq:anc143}
\end{align}
\fi
where \eqref{eq:abc132} follows because $\boldg_t$ lies in the subspace $M$, and \eqref{eq:anc143} follows because $\boldb_t\sim\mathcal{N}(0,\sigma^2 I_p)$ and thus $\mathbb{E}_{\boldb_t}\left[\snorm{\boldb_t}{M}^2\right]=\rank(M)\cdot\sigma^2$.
Rearranging the terms in \eqref{eq:anc143}, we have the following.
\ifaistats
\begin{align}
\mathbb{E}_{\boldb_1,\dots,\boldb_{t-1}}\left[\ip{\boldg_t}{\theta_t-\theta^*}\right]
&\leq\frac{1}{2\eta}\left(\Psi_{{t-1}}(\theta_t)-\Psi_{{t}}(\theta_{t+1})\right)\nonumber\\
&+\frac{\eta}{2}\left(L^2+\rank(M)\cdot\sigma^2\right).
\label{eq:abc321}
\end{align}
\else
\begin{align}
\mathbb{E}_{\boldb_1,\dots,\boldb_{t-1}}\left[\ip{\boldg_t}{\theta_t-\theta^*}\right]
\leq\frac{1}{2\eta}\left(\Psi_{{t-1}}(\theta_t)-\Psi_{{t}}(\theta_{t+1})\right)+\frac{\eta}{2}\left(L^2+\rank(M)\cdot\sigma^2\right).
\label{eq:abc321}
\end{align}
\fi
Let $\Psi(\theta)=\snorm{\theta-\theta^*}{M}^2$.
Summing up \eqref{eq:abc321} for all $t \in [T]$, averaging over the $T$ iterations, and combining with \eqref{eq:12as}, we get:
\ifaistats
\begin{align}
&\mathbb{E}\left[\calL\left(\privT;D\right)\right]-\calL\left(\theta^*;D\right)
\nonumber\\
&\leq\frac{1}{2T\eta}\Psi(\boldsymbol{0})+\frac{\eta}{2}\left(L^2+\rank(M)\cdot\sigma^2\right)
\label{eq:ls134}
\end{align}
\else
\begin{align}
\mathbb{E}\left[\calL\left(\privT;D\right)\right]-\calL\left(\theta^*;D\right)
\leq\frac{1}{2T\eta}\Psi(\boldsymbol{0})+\frac{\eta}{2}\left(L^2+\rank(M)\cdot\sigma^2\right)
\label{eq:ls134}
\end{align}
\fi
Setting $\eta$ to minimize the RHS, we have
\ifaistats
\begin{align*}
&\mathbb{E}\left[\calL\left(\privT;D\right)\right]-\calL\left(\theta^*;D\right) \\
&\leq \snorm{\theta^*}{M} \sqrt{\frac{L^2+\rank(M)\cdot\sigma^2}{T}} \\
&= \snorm{\theta^*}{M} \sqrt{\frac{L^2}{T}+\frac{2L^2 \log(1/\delta)\cdot\rank(M)}{n^2\eps^2}},
\end{align*}
\else
\begin{align*}
\mathbb{E}\left[\calL\left(\privT;D\right)\right]-\calL\left(\theta^*;D\right)
&\leq \snorm{\theta^*}{M} \sqrt{\frac{L^2+\rank(M)\cdot\sigma^2}{T}} \\
&= \snorm{\theta^*}{M} \sqrt{\frac{L^2}{T}+\frac{2L^2 \log(1/\delta)\cdot\rank(M)}{n^2\eps^2}},
\end{align*}
\fi
where the equality follows by plugging in $\sigma=\frac{L\sqrt{2T\log(1/\delta)}}{n\epsilon}$.
Now, setting $T=n^2\epsilon^2$, we have
\ifaistats
\begin{align*}
&\mathbb{E}\left[\calL\left(\privT;D\right)\right]-\calL\left(\theta^*;D\right)\\
&\leq \frac{L\snorm{\theta^*}{M}\sqrt{{1 + 2\cdot\sf rank}(M)\cdot\log(1/\delta)}}{\epsilon\cdot n}.
\end{align*} 
\else
\begin{align*}
\mathbb{E}\left[\calL\left(\privT;D\right)\right]-\calL\left(\theta^*;D\right)
\leq \frac{L\snorm{\theta^*}{M}\sqrt{1 + 2\cdot\rank(M)\cdot\log(1/\delta)}}{\epsilon\cdot n}.
\end{align*} 
\fi
This completes the proof.
\end{proof}

The lower-bound in \citet{BST14} shows that if one performs constrained optimization with DP, then the excess empirical risk is $\widetilde{\Omega}(\sqrt{p}/(\epsilon n))$. This lower bound holds true for generalized linear problems as well. However, since we consider \emph{unconstrained} optimization here, the lower bound does not apply to our result.
In fact,
the lower bound does not hold even for general convex functions, as long as the underlying optimization problem is unconstrained. Subsequent to our work, \citet{kairouz2020dimension} showed that via adaptive preconditioning methods, one can get dimension independent empirical risk bounds for general convex problems as long as the gradients lie in a low-rank subspace. However, for the specific case of GLMs, their bounds are weaker than ours, and depend on $\max\limits_{t\in T}\|\theta_t-\theta^*\|_2$, rather than $\snorm{\theta^*}{M}$ (as defined in Theorem~\ref{thm:dimIndependence}).

The guarantee in Theorem~\ref{thm:dimIndependence} is of the same flavor as in \citet{jain2014near}, wherein such a result was shown for two different DP algorithms, namely, \emph{output perturbation}, and \emph{objective perturbation} \citep{chaudhuri2011differentially, kifer2012private}. 
Our result via $\adpgd$ improves the state-of-the-art in the following ways.
First, the result in \citet{jain2014near} was only for the worst-case setting where $\rank(M)=n$, whereas our results extend to any rank setting. 
Second, unlike {output perturbation} and {objective perturbation}, $\adpgd$ \emph{does not} require {convexity to ensure DP}. As a result, $\adpgd$ can be applied to non-convex losses and may enjoy the same dimension-independent behavior as in the convex case. (We show evidence of this in Section~\ref{sec:privNonConvex}.)
Third, our results for DP-GD almost seamlessly transfer to the local differential privacy (LDP) setting\footnote{Obtaining LDP guarantee results in scaling up the Gaussian noise in Algorithm~$\adpgd$ by $\sqrt{n}$ factor.}~\citep{Warner,evfimievski2003limiting,KLNRS}. This is the \emph{first} dimension-independent excess risk guarantee in the LDP setting. Output perturbation and objective perturbation are incompatible with LDP, as they require a centralized dataset to operate. 

\mypar{Obtaining optimal excess population risk guarantee} We translate the excess empirical risk guarantee in Theorem~\ref{thm:dimIndependence} to population risk guarantee via the standard approach of uniform stability~\citep{BassilyFTT19,bassily2020stability}. The argument in Theorem~\ref{thm:dimIndPop} is identical to that in~\citep{bassily2020stability}, except that we operate with $\snorm{\cdot}{M}$ (defined in Theorem~\ref{thm:dimIndPop}) instead of the $\ell_2$-norm. Note the bound is asymptotically the same as the optimal excess population risk in \citet{BassilyFTT19}, except that the ambient dimensionality ($p$) is replaced with $\min\{\rank(M),n\}$. We provide the proof of Theorem~\ref{thm:dimIndPop} in Appendix~\ref{app:pop}.

\begin{thm}
Let $\theta_0=\mathbf{0}^p$ be the initial point of $\adpgd$.
Let $D=\{(\boldx_1,y_1),\ldots,(\boldx_n,y_n)\}$ be drawn i.i.d. from a distribution $\tau$.
Let $\thetapop=\argmin\limits_{\theta\in\bfR^p}\mathbb{E}_{(\boldx,y)\sim\tau}\left[\ell(\ip{\boldx}{\theta};y)\right]$, and $M$ be the projector to the eigenspace of the matrix $\mathbb{E}_{\bfx\sim\tau}\left[\bfx\bfx^T\right]$. 
Let $L$ be the gradient $\ell_2$-norm bound, and $k=\min\{\rank(M),n\}$. 
Setting the constraint set $\calC=\boldR^p$ and running $\adpgd$ on $\calL(\theta;D)$ for $T=n^2$ steps with appropriate learning rate $\eta$, we obtain
\ifaistats
\begin{align*}
&\mathbb{E}_{\adpgd,D,(\bfx,y)\sim\tau}\left[\ell(\ip{\privT}{\bfx};y)-\ell(\ip{\thetapop}{\bfx};y)\right]\\&
=\snorm{\thetapop}{M}\cdot O\left(\max\left\{\frac{1}{\sqrt n},\frac{L\sqrt{k\cdot\log(1/\delta)}}{\epsilon n}\right\}\right).
\end{align*}
\else
\begin{align*}
&\mathbb{E}_{\adpgd,D,(\bfx,y)\sim\tau}\left[\ell(\ip{\privT}{\bfx};y)-\ell(\ip{\thetapop}{\bfx};y)\right]=\snorm{\thetapop}{M}\cdot O\left(\max\left\{\frac{1}{\sqrt n},\frac{L\sqrt{k\cdot\log(1/\delta)}}{\epsilon n}\right\}\right).
\end{align*}
\fi
\label{thm:dimIndPop}
\end{thm}

\subsection{Lower-bound via Fingerprinting Codes}
\label{sec:lbfp}
Next, we prove the optimality of the upper bound in Theorem~\ref{thm:dimIndependence} up to a $O(\sqrt{\log(1/\delta)})$ factor\footnote{This dependence on the $\delta$ parameter can be introduced into the lower bound by a generic group privacy reduction \citep{SteinkeU15}.} by proving a matching lower bound. The lower bound is attained by a simple convex GLM given by $\ell(\ip{\theta}{\bfx};y) = |\langle \theta, \bfx \rangle - y|$; this loss is $L$-Lipschitz for $L=\|\bfx\|_2$ and is always non-negative. 

We emphasize that our lower bound holds in the unconstrained setting where $\calC = \bfR^p$. Prior lower bounds \citep{BST14} heavily relied on $\theta$ being constrained; indeed, if we remove the constraint on $\theta$ in prior lower bounds, either there is no minimum value (i.e.,~$\inf_{\theta \in \bfR^p} \calL(\theta;D) = -\infty$) or, if we avoid this with a regularizer, the loss is no longer Lipschitz.

\begin{thm}
Let $\calA : (\bfR^p \times \bfR)^n \to \bfR^p$ be an arbitrary $(\varepsilon,\delta)$-differentially private algorithm with $\varepsilon \le c$ and $\delta \le c\varepsilon/n$ for some universal constant $c>0$. Let $1 \le d \le p$ be an integer. Let $\ell : \bfR^p \times (\bfR^p \times \bfR) \to \bfR$ be defined by $\ell(\ip{\theta}{\bfx};y) = |\langle \theta, \bfx \rangle - y|$. Then there exists a data set $D = ((\bfx_1,y_1),\cdots,(\bfx_n,y_n)) \in (\bfR^p \times \bfR)^n$ such that the following holds. 
For all $i \in [n]$, $y_i \in \{0,1\}$ and $\|\bfx_i\|_2 \le 1$, $\|\partial_\theta \ell(\theta;(\bfx_i,y_i))\|_2 \le 1$ for all $\theta \in \bfR^p$. Let $\privT = \calA(D)$ and $\theta^*=\argmin\limits_{\theta\in\bfR^p}\calL(\theta;D)$ (breaking ties towards lower $\|\theta\|_2$), where $\calL(\theta;D)=\frac{1}{n} \sum_{i=1}^n \ell(\ip{\theta}{\bfx_i};y_i)$. Then we have $\rank(\sum_{i=1}^n \bfx_i \bfx_i^\top)\le d$ and $\|\theta^*\|_2 \le \sqrt{d}$ and
\ifaistats
\begin{align*}
&\mathbb{E}\left[\calL\left(\privT;D\right)\right]-\calL\left (\theta^*;D\right) \ge  \Omega\left(\min\left\{1,\frac{d}{\varepsilon n}\right\}\right) \\&~~\ge \Omega\left(\min\left\{ 1 , \frac{\|\theta^*\|_2 \cdot \sqrt{\rank\left(\sum_{i=1}^n \bfx_i \bfx_i^\top\right)}}{\varepsilon n}\right\}\right).
\end{align*}
\else
\begin{align*}
&\mathbb{E}\left[\calL\left(\privT;D\right)\right]-\calL\left (\theta^*;D\right) \ge  \Omega\left(\min\left\{1,\frac{d}{\varepsilon n}\right\}\right) \ge \Omega\left(\min\left\{ 1 , \frac{\|\theta^*\|_2 \cdot \sqrt{\rank\left(\sum_{i=1}^n \bfx_i \bfx_i^\top\right)}}{\varepsilon n}\right\}\right).
\end{align*}
\fi
\label{thm:ERM-lb}
\end{thm}
The parameter $d$ determines the complexity of the lower bound instance and controls the rank of the feature matrix. %

We now sketch the proof of Theorem \ref{thm:ERM-lb}. The full details are in \ifsupp Appendix \ref{sec:mspdim}.
\else the supplementary material. \fi
The proof is based on the powerful fingerprinting technique for proving differential privacy lower bounds \citep{bun2018fingerprinting,SteinkeU17,DworkSSUV15}, which was originally developed in the cryptography literature \citep{BonehS98,Tardos08}.
A fingerprinting code provides a sequence of vectors $\bfz_1,\cdots,\bfz_n \in \{0,1\}^d$, which we use to construct the hard dataset with each $\bfz_i$ corresponding to the data of individual $i$. The guarantee of the fingerprinting code is that it is not possible to privately estimate the average $\frac{1}{n} \sum_{i=1}^n \bfz_i \in [0,1]^d$ to within a certain level of accuracy (depending on the privacy parameters $\varepsilon$ and $\delta$, the dimensionality $d$, and the number of individuals $n$).

We construct the hard dataset by setting $y_i = \langle \bfz_i , \bfx_i \rangle$ for all $i \in [n]$. Intuitively, to obtain low loss $\calL(\privT;D) = \frac{1}{n} \sum_{i=1}^n |\langle \privT , \bfx_i \rangle - y_i| = \frac{1}{n} \sum_{i=1}^n |\langle \privT - \bfz_i , \bfx_i \rangle |$ the algorithm must ensure $\privT \approx \frac{1}{n} \sum_{i=1}^n \bfz_i$, but this is impossible to do privately due to the properties of the fingerprinting code. This is the basis of our lower bound.

The only remaining part the construction is the feature vectors $\bfx_1, \cdots, \bfx_n \in \{0,1\}^d$. These are either standard basis vectors (i.e., one $1$) or all zeros.  These are chosen so that $ \frac{1}{n} \sum_{i=1}^n |\langle \privT - \bfz_i , \bfx_i \rangle | = \Theta(\| \privT - \frac{1}{n} \sum_{i=1}^n \bfz_i \|_1)$, which suffices for the above argument.
\label{sec:lbdpgd}

Independently and subsequent to our work, \citet{CaiWZ20} prove lower bounds for estimating the parameters of a GLM. In contrast, our lower bound is stated in terms of the loss. They also give dimension-dependent upper bounds.
\section{Dimension Independent First-order Convergence for Non-convex GLMs}
\label{sec:privNonConvex}

In this section, we provide an extension to Theorem \ref{thm:dimIndependence} that captures the setting when the loss function 
$\ell(z;\cdot)$ may be non-convex in $z$. Such loss functions appear commonly in robust regression, such as Savage loss~\citep{masnadi2009design}, Tangent loss~\citep{masnadi2010design}, and tempered loss~\citep{amid2019robust}. 
We show that as long as $\ell(z;\cdot)$ is $\beta$-smooth (see Definition~\ref{def:smooth}), 
$\adpgd$ (Algorithm~\ref{Alg:Frag}) approximately reaches a stationary point on the objective function $\calL(\theta;D)$, where $\theta$ is called a stationary point if $\nabla\calL(\theta;D)=0$. Similar to Theorem~\ref{thm:dimIndependence}, the convergence guarantee in this section will have no explicit dependence on the number of dimensions. We use a folklore argument stated in \citet{allen2018natasha} to prove our result. The proof of Theorem~\ref{thm:util31} is in \ifsupp Appendix~\ref{app:nonconvex}. \else the supplementary material. \fi

\begin{thm}
Recall the notation in Theorem \ref{thm:dimIndependence}. Let $\Theta = [\theta_0,\ldots,\theta_T]$ be the list of models output by $\adpgd$. Let $t_{\sf priv}\leftarrow \argmin\limits_{t\in \{0,\dots,T-1\}}\snorm{\nabla\calL(\theta_t;D)}{M}+{\sf Lap}\left(\frac{4L}{n}\right)$. Then, the algorithm that outputs the set $\Theta$ in conjunction with $t_{\sf priv}$ is $(2\epsilon,\delta)$-differentially private. Furthermore, as long as $T\geq \frac{\beta n^2\epsilon^2\cdot\calL(\mathbf{0}^p;D)}{2L^2\log\left(\frac{1}{\delta}\right)}$, we have with probability at least $1-\gamma$,
\ifaistats
\begin{align*}
&  \snorm{\nabla\calL(\theta_{t_{\sf priv}};D)}{2} 
= \snorm{\nabla\calL(\theta_{t_{\sf priv}};D)}{M} \\
&= O\left(\frac{L}{\epsilon n}\cdot\sqrt{\rank(M)\cdot\log\left(\frac{1}{\delta}\right)\log\left(\frac{T}{\gamma}\right)}\right).    
\end{align*}
\else
\begin{align*}
  \snorm{\nabla\calL(\theta_{t_{\sf priv}};D)}{2} 
= \snorm{\nabla\calL(\theta_{t_{\sf priv}};D)}{M}
= O\left(\frac{L}{\epsilon n}\cdot\sqrt{\rank(M)\cdot\log\left(\frac{1}{\delta}\right)\log\left(\frac{T}{\gamma}\right)}\right).    
\end{align*}
\fi
Here, $L$ is the Lipschitz constant, $\beta$ is the smoothness constant of $\calL(\theta;D)$. We set a constant learning rate in Algorithm \ref{Alg:Frag} as $\eta=\frac{1}{\beta}$, and $\theta_0=\boldsymbol{0}^p$. 
Notice that $\rank(M)\leq n$ always holds but $\rank(M)$ can be much smaller than $n$.
\label{thm:util31}
\end{thm}

The algorithm in Theorem~\ref{thm:util31} is a modification of $\adpgd$ (Algorithm~\ref{Alg:Frag}), where from the list of models output by $\adpgd$ $\theta_0,\ldots,\theta_T$, we select the one with approximately minimum gradient norm for the loss  $\calL(\theta;D)$ via report-noisy-max (a.k.a. the exponential mechanism)~\citep{mcsherry2007mechanism,dwork2014algorithmic}.

Theorem~\ref{thm:util31} does not immediately imply convergence to a local minima or a bound on the population risk. However, it demonstrates that convergence of DP-GD  can be dimension-independent even in the case of non-convex losses. It is perceivable that this line of argument be extended for convergence to a local minimum using techniques similar to those in \citet{jin2017escape}. However, that would require an additional assumption beyond smoothness, i.e., \emph{Lipschitz continuity} of the Hessian.

\section{Dimension-Independent Convergence of Clipped Private Gradient Descent}
\label{sec:clipping}

In Sections~\ref{sec:ucconvex} and~\ref{sec:privNonConvex}, we have seen that $\adpgd$ (Algorithm~\ref{Alg:Frag}) has dimension-independent convergence for unconstrained GLMs (convex or non-convex). Moreover, the bound is tight in the convex case. 
An issue with $\adpgd$ is that it requires an apriori bound on the $\ell_2$-norm of the gradients for all values of $\theta$ and $d$. In practice~\citep{DP-DL,papernot2020tempered}, when such a bound is unknown or non-existent, a modified version of $\adpgd$, called \emph{clipped} DP-GD, is used. 
In this algorithm, the gradients on individual data points are projected to ensure that they conform to a pre-defined $\ell_2$-norm bound (denoted by \emph{clipping norm} $\cn$). This is done by  replacing Line 3 in Algorithm~\ref{Alg:Frag} by
\begin{align}
\gpriv_t \leftarrow &\frac{1}{n}\sum\limits_{i=1}^n \clip\left(\subg \ell(\theta_t;d_i)\right) + \calN\left(0,\sigma^2\right)
\label{eq:clipGD}
\end{align}
where $\clip(\boldv)=\boldv\cdot\min\left\{1,\frac{\cn}{\|\boldv\|_2}\right\}.$ We denote this algorithm as $\adpgdc$.

Now, we show that $\adpgdc$ enjoys the same dimension-independent convergence as $\adpgd$ as long as the underlying problem is \emph{convex} GLM. For all values of {clipping norm} $\cn>0$, the convergence is always to the minimum of a well-defined convex objective function. Based on the value of $\cn$ (i.e., whether $\cn \leq L$), the convergence may be to the minimum of a different (well-specified) objective function, rather than the original objective $\calL(\theta;D)$.

\subsection{Characterizing Clipping in Convex GLMs}
\label{sec:clipeqvhuber}

First, we provide an analytical tool in Lemma \ref{lem:huberized} to precisely quantify the objective function that $\adpgdc$ optimizes when the underlying loss function is a convex GLM.

\begin{lem}\label{lem:truncation}\label{lem:huberized}
Let $f:\bfR\to\bfR$ be any convex function and $\cn \in \bfR_+$ be any positive value.
For any $\bfx \neq \mathbf{0}$, let
$Y_1 = \bbracket{y: u < -\frac{\cn}{\snorm{\boldx}{2}} \ \forall u \in \partial f(y)}$ and 
$Y_2 = \bbracket{y: u > \frac{\cn}{\snorm{\boldx}{2}} \ \forall u \in \partial f(y)}$.
If $Y_1$ is non-empty, let $y_1 = \sup Y_1$; otherwise  $y_1 = -\infty$.
If $Y_2$ is non-empty, let $y_2 = \inf Y_2$; otherwise  $y_2 = \infty$.
Let $g_{\boldx}:\bfR\to\bfR$ be
\begin{align*}
g_{\boldx}(y) = 
\begin{cases}
-\frac{\cn}{\|\bfx\|_2} (y - y_1) + f(y_1) & \text{for } y \in (-\infty, y_1) \\
f(y) & \text{for } y\in [y_1, y_2] \cap \bfR \\
\frac{\cn}{\|\bfx\|_2} (y - y_2) + f(y_2) & \text{for } y \in (y_2,\infty)
\end{cases}.
\end{align*}
Then the following holds.
\begin{enumerate}[leftmargin=12pt]
\item $g_{\boldx}$ is convex. 
\item 
Let $\ell_f:\bfR^p\times\bfR^p\to\bfR$ be $\ell_f(\theta;\bfx)=f\left(\ip{\theta}{\bfx}\right)$ for any $\theta,\bfx$. 
Let $\ell_g:\bfR^p\times\bfR^p\to\bfR$ be $\ell_g(\theta; \bfx) = g_{\boldx}(\ip{\theta}{\bfx})$ for any $\theta,\bfx$.
Then, for any $\theta$, $\boldx$, we have
\begin{align*}
\partial_{\theta}\ell_g(\theta;\bfx) = 
\bbracket{\min\left\{1,\frac{\cn}{\snorm{\boldu}{2}}\right\}\cdot \boldu : \boldu \in \partial_{\theta}\ell_f(\theta;\bfx)}.
\end{align*}
\end{enumerate}
\end{lem}

\mypar{Note} Lemma \ref{lem:huberized} is a generic tool for understanding the effect of clipping. In fact, it can be used to justify the use of standard private convex optimization analysis in \citep{BST14,FMTT18,feldman2019private,bassily2019private} to $\adpgdc$ on convex GLMs.

The proof (see \ifsupp Appendix~\ref{app:undclip}\else the supplementary material\fi) is based on the fact that clipping does not affect the monotonicity property of the derivative of one-dimensional convex function. A pictorial representation of Lemma~\ref{lem:huberized} is given in Figure~\ref{fig:clipnormlr}.

\begin{figure}[h]
\centering
\includegraphics[width=0.45\textwidth]{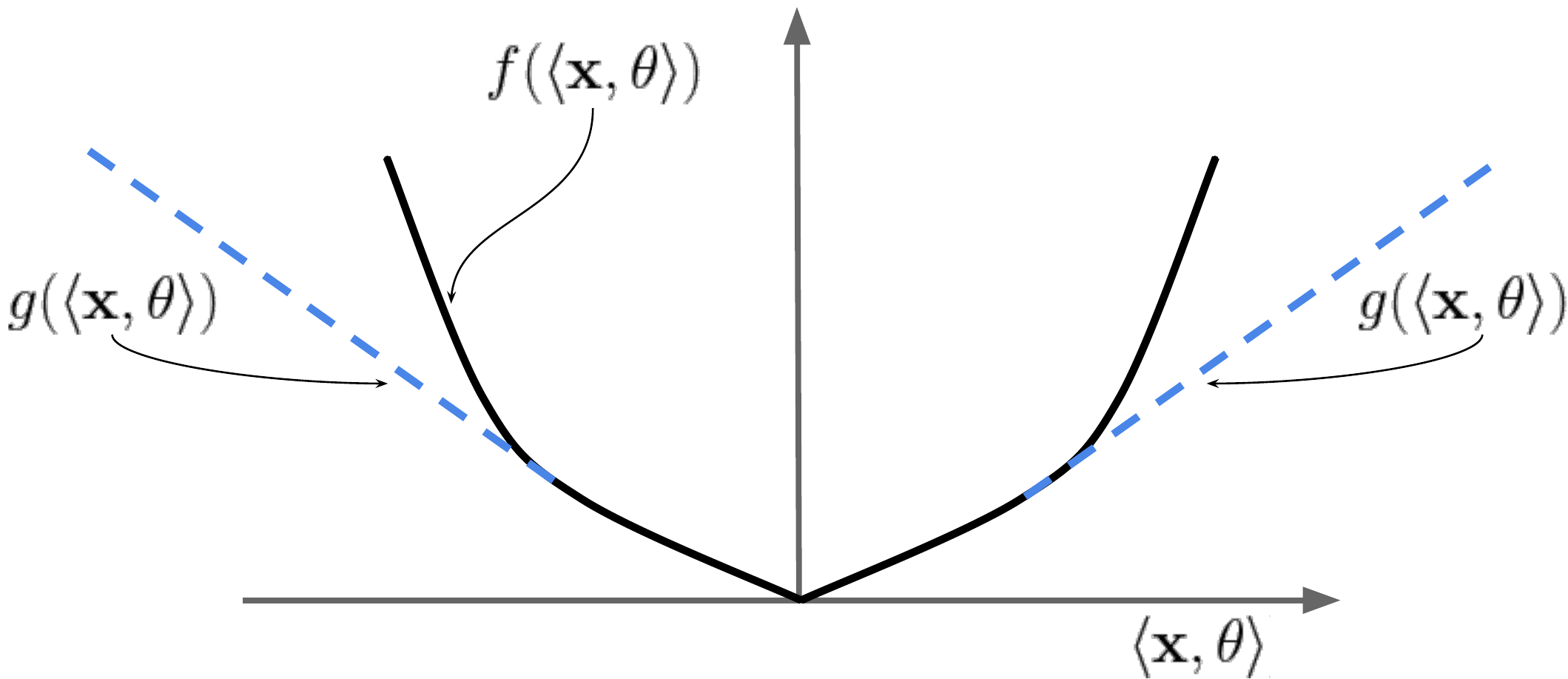}
\caption{Representing the clipping operation in~\eqref{eq:clipGD} for Algorithm $\adpgdc$ via Lemma~\ref{lem:truncation}.}
\label{fig:clipnormlr}
\end{figure}

\subsection{Rank-based Convergence of Algorithm $\adpgdc$}

Now, we provide the rank-based convergence of Algorithm $\adpgdc$. Recall that the original loss function is given by $\calL(\theta;D)=\frac{1}{n}\sum\limits_{i=1}^n\ell(\ip{\theta}{\boldx_i};y_i)$. Since, by Lemma~\ref{lem:huberized} we know that for any clipping level $\cn>0$, one can replace each loss function $\ell(\ip{\theta}{\boldx_i};y_i)$ with that obtained from Lemma~\ref{lem:huberized}. We denote these clipped losses as $\ellc(\ip{\theta}{\boldx_i};y_i)$, and the corresponding objective function as $\calLc(\theta;D) =\frac{1}{n}\sum\limits_{i=1}^n\ellc(\ip{\theta}{\boldx_i};y_i)$. We immediately have the following corollary from Theorem~\ref{thm:dimIndependence}.

\begin{cor}
Let $\theta_0=\mathbf{0}^p$ be the initial point for $\adpgdc$, and let $\cn>0$ be the corresponding clipping norm. 
Let $\thetaH = \argmin\limits_{\theta\in\bfR^p}\calLc(\theta;D)$, and $M$ be the projector to the eigenspace of matrix $\sum\limits_{i=1}^n \bfx_i\bfx_i^T$. For constraint set $\calC=\boldR^p$, clipping norm $L$, and running $\adpgdc$ on $\calL(\theta;D)$ for $T=n^2\epsilon^2$ steps with appropriate learning rate $\eta$, we get
\ifaistats
\begin{align*}
& \mathbb{E}\left[\calLc\left(\privT;D\right)\right]-\calLc\left (\thetaH;D\right)\\
& \qquad \leq\frac{L\snorm{\thetaH}{M}\sqrt{1 + 2\cdot\rank(M)\cdot \log(1/\delta)}}{\epsilon n}.
\end{align*}
\else
\begin{align*}
\mathbb{E}\left[\calLc\left(\privT;D\right)\right]-\calLc\left (\thetaH;D\right)
\leq\frac{L\snorm{\thetaH}{M}\sqrt{1 + 2\cdot\rank(M)\cdot \log(1/\delta)}}{\epsilon n}.
\end{align*}
\fi
Here, $\rank(M)\leq n$ (but can be much smaller), and $\snorm{\cdot}{M}$ is the seminorm w.r.t. the projector $M$.
\label{cor:dimIndependencec}
\end{cor}

From Corollary~\ref{cor:dimIndependencec}, it is immediate that as long as the clipping norm $\cn \geq L$ (where $L$ is the $\ell_2$-Lipschitz constant of the loss functions $\ell$), $\adpgdc$ optimizes the original loss $\calL$. 

\subsection{Interlude: Adverse Effects of Aggressive Clipping}
\label{sec:negclip}

Corollary~\ref{cor:dimIndependencec} eludes two natural questions: 
i) Is the output of Algorithm $\adpgdc$ guaranteed to be good in terms of the empirical loss on the original loss $\calL(\theta;D)$ when $\cn < L$, and 
ii) When the underlying loss  \emph{does not} satisfy the convex GLM property, is the objective function that gets optimized by $\adpgdc$ still well-defined as for convex GLMs?
Our answers to both questions are \emph{negative}.

To answer the first question, we construct an instance of logistic regression $\calL(\theta;D)$ with $L$, the $\ell_2$-Lipschitz constant of the individual loss functions, being $1$. We show that if the clipping norm $\cn < 1/4$, then the excess empirical risk of Algorithm $\adpgdc$ on $\calL(\theta;D)$ is $\Omega(1)$. Whereas Algorithm $\adpgd$ with $L=1$ has an excess empirical risk of $O(1/n)$, where $n$ is the number of training examples. A formal description of this bound is in \ifsupp Appendix~\ref{sec:lower bound logistic regression}. \else the supplementary material. \fi

To answer the second question, we show that even for simple convex problems like multi-class logistic regression (that does not belong to the traditional GLM class), clipping in $\adpgdc$ can generate a sequence of ``clipped gradient vectors'' that do not conform to a gradient field of any ``natural'' convex function, 
though the output of $\adpgdc$ might still  minimize the excess empirical risk on $\calL(\theta;D)$ practically. 
We leave this exploration for future work.

Formally, consider a $K$-class classification problem for $K \geq 3$.
Given a sample $(x,y)$ with $x\in\bfR^p$ and $y\in [K]$, the cross-entropy loss $\ell: \bfR^{p\times K} \times \bfR^p\times [K] \to \bfR$ is, for $\bt = [\theta^{(1)}, \dots, \theta^{(K)}]$,
\begin{align*}
\ell\bracket{\bt; (x,y)} = \sum_{k=1}^K \indi{y=k} \log \frac{\expo{\theta^{(k)}\cdot x}}{\sum_{k'=1}^K \expo{\theta^{(k')}\cdot x}}.
\end{align*}

\begin{thm}
Consider any sample $(x, y)$ with $x\in \bfR^p\backslash{\{\mathbf{0}\}}$, $y\in [K]$ (for $K \geq 3$) and any $\cn > 0$ such that $\bT = \bbracket{\bt: \|\nabla_{\bt} \ell(\bt;(x,y))\|_2 > \cn}$ is non-empty.
Let $G(\theta)$ be the clipped gradient of $\ell(\bt;(x,y))$ as defined above. 
Consider any function $f:\calC \to \bfR$, $\calC\subseteq \bfR^{p\times K}$ such that $\interior{\calC}\cap \bT\neq \emptyset$, where $\interior{\calC}$ is the interior of set $\calC$.
If $f$ is differentiable everywhere except for a set $\calC_N\subseteq \calC$
s.t. $\calC_N$ is a closed set on $\calC$ and has zero Lebesgue measure, then it is not possible for $\nabla_{\bt} f(\bt) = G(\bt)$ to hold for all $\bt \in \interior{\calC}\backslash\calC_N$.
\label{thm:clipped_softmax}
\end{thm}

As convexity implies differentiability almost everywhere~\cite[Theorem 25.5]{rockafellar1970convex}, if $f$ is convex, we only need $\calC_N$ to be a closed set. 
The $\ell_1$ regularizer $\snorm{\bt}{1}$ is non-differentiable on a closed set, and so is the hinge loss $\ell_{\sf hinge}(\theta; (x,y))=\max\bracket{0, 1-y\ip{\bt}{x}}$.
Theorem~\ref{thm:clipped_softmax} essentially rules out the possibility that the field of clipped gradients corresponds to any single objective function in convex models like softmax regression and SVMs with $\ell_1/\ell_2$ regularization, and in non-convex models like neural networks with popular activation functions like logistic sigmoid, tanh, ReLu.
We defer the proof of Theorem~\ref{thm:clipped_softmax} to \ifsupp Appendix~\ref{app:softMax}.
\else the supplementary material. \fi
One might further ask if the problem above could be resolved by ``per-class'' clipping, i.e., clipping $\nabla_{\theta^{(k)}} \ell$ individually for each $k$? The answer is still negative. 
We provide more details in \ifsupp Appendix~\ref{sec:per-class_clip}. \else the supplementary material. \fi

\section{Experiments}
\label{sec:expts}

We conduct experiments to demonstrate the dimension-independence of DP-GD with Gaussian noise. The setup mainly follows that from \citet[Figure 1(c)]{jain2014near}. 
We use a normalized version of the Cod-RNA dataset~\citep{codrna}, and use logistic regression to solve for binary classification.
Following~\citet{jain2014near}, we append zero-valued features to the data samples, such that the accuracy of a non-private classifier does not change. 
To solve the problem privately, we consider DP-SGD with mini-batch gradient and Gaussian noise. 
For comparison, we consider DP-SGD with noise drawn from a Gamma distribution. 
We fix $\epsilon \approx 5.0$ and $\delta=10^{-5}$.
Additional details on the setup can be found in Appendix~\ref{app:expts}.

In Figure~\ref{fig:expt}, we report the final test accuracy under different data dimensionality after tuning the learning rate. 
It is clear that DP-SGD with Gaussian noise is dimension-independent, as is predicted by Theorem~\ref{thm:dimIndependence}.
On the other hand, the accuracy of DP-SGD with Gamma noise decreases as data dimensionality increases.

\begin{figure}[ht]
\centering
\includegraphics[width=0.39\textwidth]{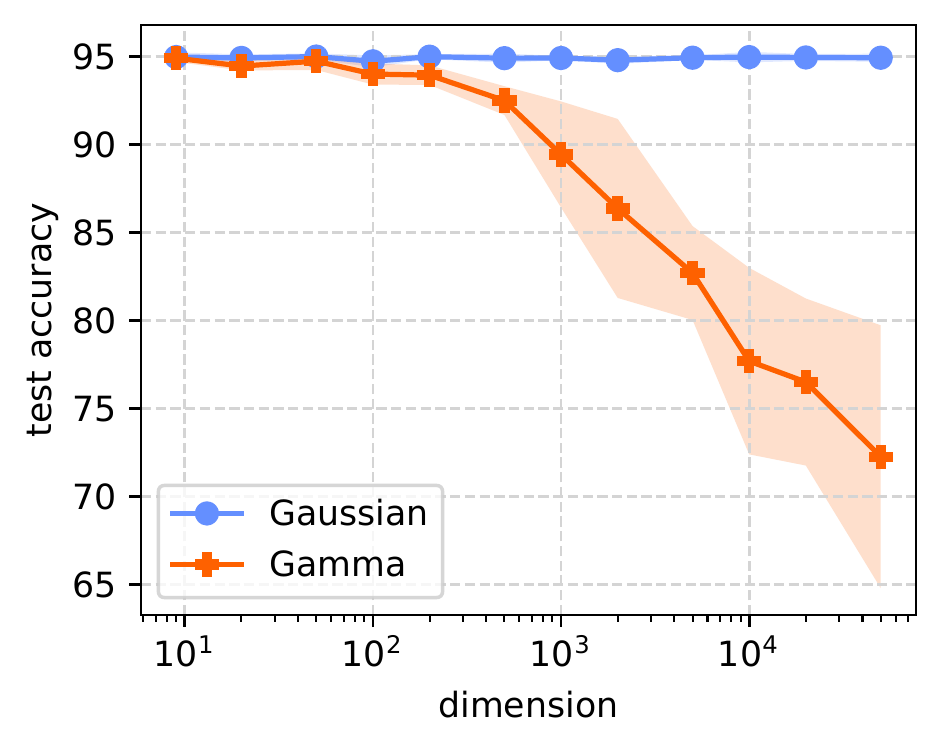}
\caption{Test accuracy vs. data dimensionality for DP-SGD with i) Gaussian noise, and ii) Gamma noise. 
}
\label{fig:expt}
\end{figure}

\section*{Acknowledgements}
\label{sec:acks}
The authors would like to thank Vitaly Feldman, Vivek Kulkarni, Monica Ribero, Shadab Shaikh, Adam Smith, Kunal Talwar, Andreas Terzis, and Jalaj Upadhyay for their helpful comments towards improving the paper.

\clearpage

\bibliographystyle{abbrvnat}
\bibliography{reference}
\ifsupp
\onecolumn
\appendix
\section{Proof of Excess Population Risk in Theorem~\ref{thm:dimIndPop}}
\label{app:pop}

The proof will closely follow the arguments in Theorem 6.2 in~\citep{bassily2020stability}. We first provide Definition~\ref{def:uas} and Lemma~\ref{lem:approxStab} which are crucial to the proof.

\begin{defn}[Uniform argument stability (UAS)~\citep{bassily2020stability}] Given a (randomized) algorithm $\calA$, and two data sets $D$ and $D'$ differing in exactly one row, we define uniform argument stability (w.r.t. a semi-norm $\snorm{\cdot}{*}$) as $\gamma_\calA\left(D,D',\snorm{\cdot}{*}\right)=\snorm{\calA(D)-\calA(D')}{*}$. Here the random variables $\calA(D)$ and $\calA(D')$ have the same choice of random coins.
\label{def:uas}
\end{defn}

\begin{lem} Consider any data set $D=\left\{(\bfx_1,y_1),\ldots,(\bfx_n,y_n)\right\}$ and its neighbor $D'$, where for some $j\in[n]$, the $j$-th sample is replaced by some $(\bfx'_j,y'_j)$ from the domain. Consider any GLM loss on the data set defined as $\calL(\theta;D)=\frac{1}{n}\sum\limits_{i=1}^n\ell(\ip{\theta}{\bfx_i};y_i)$, 
such that the $\ell_2$-norm of the gradient of $\ell(\cdot;\cdot)$ (w.r.t. $\theta$) is upper bounded by $L$.
Let $M$ be any projector that spans all the feature vectors $\left\{\bfx_i\right\}_{i=1}^n\cup\{\bfx'_j\}$.
Then, Algotithm~\ref{Alg:Frag} (Algorithm $\adpgd$) has UAS of
$$
\gamma_{\adpgd}\left(D,D',\snorm{\cdot}{M}\right)
\leq 4L\eta\left(\frac{T}{n}+\sqrt{T}\right).$$
\label{lem:approxStab}
\end{lem}

\begin{proof}
Let $\{\theta_1,\ldots,\theta_T\}$ be the outputs at each step of Algorithm $\adpgd$ on data set $D$, and $\{\theta'_1,\ldots,\theta'_T\}$ be that on data set $D'$. We first show that for all $t\in[T]$,
\begin{equation}
\snorm{\theta_t-\theta'_t}{M}\leq 4L\eta\sqrt{t}+ \frac{4\eta Lt}{n}\label{eq:8hu}.
\end{equation}
Let $\mu_t=\snorm{\theta_t-\theta'_t}{M}$. By definition of the update step in Algorithm $\adpgd$, we have
\begin{align}
\mu^2_{t+1}&=\snorm{\left(\theta_t-\eta\nabla\calL(\theta_t;D)\right)-\left(\theta'_t-\eta\nabla\calL(\theta'_t;D')\right)}{M}^2\nonumber\\
&=\snorm{\left(\theta_t-\theta'_t\right)-\eta\left(\nabla\calL(\theta_t;D)-\nabla\calL(\theta'_t;D')\right)}{M}^2\nonumber\\
&=\snorm{\theta_t-\theta_t'}{M}^2+\eta^2\snorm{\nabla\calL(\theta_t;D)-\nabla\calL(\theta'_t;D')}{M}^2-2\eta\ip{\theta_t-\theta_t'}{\nabla\calL(\theta_t;D)-\nabla\calL(\theta'_t;D')}\nonumber\\
&=\snorm{\theta_t-\theta_t'}{M}^2+\eta^2\snorm{\nabla\calL(\theta_t;D)-\nabla\calL(\theta'_t;D')}{M}^2\nonumber\\
&\quad -2\eta\ip{\theta_t-\theta_t'}{\nabla\calL(\theta_t;D)-\nabla\calL(\theta_t;D')}-2\eta\ip{\theta_t-\theta_t'}{\nabla\calL(\theta_t;D')-\nabla\calL(\theta'_t;D')}\nonumber\\
&\leq\mu^2_t+4\eta^2L^2+ \frac{4\eta L\mu_t}{n} \nonumber\\
&=4\eta^2L^2 (t+1) + \frac{4\eta L}{n}\sum\limits_{i=1}^t\mu_i.
\label{eq:lkf43}
\end{align}
Here, we used the fact that both $\nabla\calL(\theta;D)$ and $\nabla\calL(\theta;D')$ lie in the subspace spanned by $M$. In the last inequality, we used the monotonicity property of subgradient, i.e., $\ip{\theta_t-\theta_t'}{\nabla\calL(\theta_t;D')-\nabla\calL(\theta'_t;D')}\geq 0$. Now,~\eqref{eq:8hu} follows from the same inductive argument in Lemma 3.1 in~\citep{bassily2020stability}. 

By triangle inequality, we have $\snorm{\frac{1}{T}\sum\limits_{t=1}^T\left(\theta_t-\theta_t'\right)}{M}\leq \frac{1}{T}\sum\limits_{t=1}^T\snorm{\theta_t-\theta_t'}{M}$. Plugging in the bound from~\eqref{eq:8hu} completes the proof.
\end{proof}

With Lemma~\ref{lem:approxStab} in hand we are ready to prove Theorem~\ref{thm:dimIndPop}.

\begin{proof}[Proof of Theorem~\ref{thm:dimIndPop}]
Consider the learning rate $\eta=\frac{\snorm{\thetapop}{M}}{Ln\cdot\max\left\{\sqrt n,\frac{\sqrt{k\log(1/\delta)}}{\epsilon}\right\}}$. Let $M_D$ be the projector to the eigenspace of the matrix $\sum\limits_{i=1}^n\bfx_i\bfx_i^T$. From~\eqref{eq:ls134} in the proof of Theorem~\ref{thm:dimIndependence} and $\sigma = L\sqrt{2\log(1/\delta)} / \eps$, we have 
\begin{align}
\opt=\mathbb{E}_{\adpgd}\left[\calL\left(\privT;D\right)\right]-\calL\left (\thetapop;D\right)&=\frac{\snorm{\thetapop}{M_D}^2}{2n^2\eta}+\frac{\eta L^2}{2}+\frac{\eta L^2\rank(M_D)\ln(1/\delta)}{\epsilon^2}\nonumber\\
&\leq \frac{\snorm{\thetapop}{M}^2}{2n^2\eta}+\frac{\eta L^2}{2}+\frac{\eta L^2k\ln(1/\delta)}{\epsilon^2},
\label{eq:erma}
\end{align}
where the inequality follows from the fact that $M_D$ spans a strict subspace of that spanned by $M$, and $\rank(M_D)\leq \min\{n,\rank(M)\}=k$. Plugging in the learning rate as above, we have $\opt=L\snorm{\thetapop}{M}\cdot O\left(\frac{1}{\sqrt n},\frac{\sqrt{k\log(1/\delta)}}{\epsilon n}\right)$.

We now have the following from Theorem 2.2. in~\citep{hardt2016train} and Lemma~\ref{lem:approxStab} above. Since $M$ spans the complete domain of data elements, one can use $\snorm{\cdot}{M}$ in Lemma~\ref{lem:approxStab}. Setting $T=n^2$, we have
\begin{align}
\gen=\mathbb{E}_{\adpgd,D}\left[\mathbb{E}_{(\bfx,y)\sim\tau}\left[\ell(\ip{\bfx}{\privT};y)\right]-\calL(\privT;D)\right]\leq 8Ln\eta
=\snorm{\thetapop}{M}\cdot O\left(\frac{1}{\sqrt n}\right).\label{eq:absc1}
\end{align}
Notice that $\mathbb{E}_D\left[\calL(\thetapop;D)\right]-\mathbb{E}_{(\bfx,y)\sim\tau}\left[\ell(\ip{\bfx}{\thetapop};y)\right]= 0$. Therefore, combining~\eqref{eq:erma} and~\eqref{eq:absc1}, we have 
\begin{align*}
\mathbb{E}_{\adpgd,D,(\bfx,y)\sim\tau}\left[\ell(\ip{\privT}{\bfx};y)-\ell(\ip{\thetapop}{\bfx};y)\right]&\leq \gen+\mathbb{E}_{D}\left[\opt\right]=\snorm{\thetapop}{M}\cdot O\left(\max\left\{\frac{1}{\sqrt n},\frac{L\sqrt{k\log(1/\delta)}}{\epsilon}\right\}\right).
\end{align*}
This completes the proof.
\end{proof}

\section{Proof of Lower Bound from Section~\ref{sec:lbfp}}
\label{sec:mspdim}
We prove our lower bound in Theorem~\ref{thm:ERM-lb}. We start with a distributional form of the result.

\newcommand{\ex}[2]{{\ifx&#1& \mathbb{E} \else
\underset{#1}{\mathbb{E}} \fi \left[#2\right]}}
\newcommand{\pr}[2]{{\ifx&#1& \mathbb{P} \else
\underset{#1}{\mathbb{P}} \fi \left[#2\right]}}
\newcommand{\var}[2]{{\ifx&#1& \mathsf{Var} \else
\underset{#1}{\mathsf{Var}} \fi \left[#2\right]}}
\newcommand{\dr}[3]{\mathrm{D}_{#1}\left(#2\middle\|#3\right)}
\newcommand{\Reals}{\bfR}

\begin{thm} \label{thm:lb-core}
For $p \in [0,1]^d$ and $\alpha \in [0,1]$, define a distribution $\mathcal{D}_{p,\alpha}$ on $\{0,1\}^d \times \{0,1\}$ as follows. First, with probability $1-\alpha$ we set $X=\mathbf{0} \in \{0,1\}^d$. Otherwise (i.e., with probability $\alpha$) $X \in \{(1,0,\cdots,0),(0,1,0,\cdots,0),\cdots,(0,\cdots,0,1)\} \subset \{0,1\}^d$ is a uniformly random standard basis vector. Independently $Z \in \{0,1\}^d$ is drawn from a product distribution with $\ex{}{Z}=p$. Then $Y=\langle Z, X \rangle \in \{0,1\}$. The pair $(X,Y)$ constitutes the sample from $\mathcal{D}_{p,\alpha}$.
Let $M : \left( \Reals^d \times \Reals \right)^n \to \Reals^d$ be an $(\varepsilon,\delta)$-differentially private algorithm and let $\alpha\in(0,1]$. Then there exists some $p \in [0,1]^d$ such that the following holds. 
Let $(X_0,Y_0), (X_1,Y_1),\cdots,(X_n,Y_n)$ be independent draws from $\mathcal{D}_{p,\alpha}$ and set $D=((X_1,Y_1),\cdots,(X_n,Y_n))$. If $n \le \frac{d}{10(e^\varepsilon-1)\alpha}$ and $\delta \le \frac{e^\varepsilon-1}{1000}$, then \[\ex{}{|\langle M(D), X_0 \rangle - Y_0|} - \ex{}{|\langle p , X_0 \rangle - Y_0|} \ge \frac{\alpha}{10}.\]
\end{thm}
Note that $\inf_{\theta \in \Reals^d} \ex{}{|\langle \theta , X_0 \rangle - Y_0|} \le \ex{}{|\langle p , X_0 \rangle - Y_0|}$ and we can always set $\alpha = \min \left\{ 1 , \frac{d}{10(e^\varepsilon-1)n}\right\}$ to obtain the bound \[\ex{}{|\langle M(D), X_0 \rangle - Y_0|} - \inf_{\theta \in \Reals^d} \ex{}{|\langle \theta , X_0 \rangle - Y_0|} \ge  \min \left\{ \frac{1}{10} , \frac{d}{100(e^\varepsilon-1)n}\right\}.\]
Note that the lower bound does not grow as $\delta \to 0$. Such a dependence can be introduced via a group privacy reduction \citep{SteinkeU15}.
\begin{proof}
This proof heavily builds on the lower bounds of Steinke and Ullman \citep{SteinkeU17}.

Let $\mathcal{D}_{p,\alpha}$ be as in the theorem statement. To show that there exists some $p \in [0,1]^d$ satisfying the theorem, we show that for a random $p$ the theorem holds.  It follows that some specific $p$ must satisfy the theorem, but that specific $p$ will depend on $M$ (even though the distribution does not).

We now specify the distribution. Let $\beta>0$ be a parameter to be determined later. Let $P_1, \cdots, P_d$ be independent draws from $\mathsf{Beta}(\beta,\beta)$. We will show that, with this random $P$, the distribution $\mathcal{D}_{P,\alpha}$ satisfies the criteria of the theorem.

Specifically, we first sample $P$ as above. Then we sample $(X_0,Y_0), (X_1,Y_1),\cdots,(X_n,Y_n)$ independently (conditioned on $P$) from $\mathcal{D}_{P,\alpha}$. We set $D=((X_1,Y_1),\cdots,(X_n,Y_n))$ and $\theta= \mathsf{Proj}_{[0,1]^d}(M(D))$. That is, $\theta$ is the output of the algorithm $M$ on the dataset $D$ truncated to $[0,1]^d$. We let $Z_0,Z_1,\cdots,Z_n \in\{0,1\}^d$ be the corresponding variables hidden in $\mathcal{D}_{P,\alpha}$, so that $Y_i=\langle Z_i , X_i \rangle$ for all $i \in \{0,1,\cdots,n\}$. Note that $(X_0,Y_0)$ is \emph{not} included in $D$.

We have
\begin{align*}
    \ex{}{\|\theta-P\|_1} &\le \ex{}{\|M(D)-Z_0\|_1} + \ex{}{\|Z_0-P\|_1}\\
    &= \sum_j^d \ex{}{|M(D)_j-Z_{0,j}|}  + \ex{}{|Z_{0,j}-P_j|}\\
    &= \sum_j^d \ex{}{|\langle M(D) - Z_0 , X_0 \rangle| \mid X_{0,j}=1} + \ex{}{P_j|1-P_j| + (1-P_j)|0-P_j|}\\
    &= \frac{d}{\alpha} \cdot \ex{}{|\langle M(D) - Z_0 , X_0 \rangle|} + \sum_j^d 2 \ex{}{P_j(1-P_j)}\\
    &= \frac{d}{\alpha} \cdot \ex{}{|\langle M(D) , X_0 \rangle - Y_0|} + d \cdot \frac{\beta}{1+2\beta}
\end{align*}
and
\begin{align*}
    \inf_{\theta \in \Reals^d} \ex{}{|\langle \theta , X_0 \rangle - Y_0|} &\le \ex{}{|\langle P , X_0 \rangle - Y_0|} \\
    &= \ex{}{|\langle P - Z_0 , X_0 \rangle|}\\
    &= \frac{\alpha}{d} \ex{}{\| P - Z_0 \|_1}\\
    &= \frac{\alpha}{d} \sum_j^d \ex{}{2P_j(1-P_j)}\\
    &= \frac{\alpha\beta}{1+2\beta}.
\end{align*}
Hence \[\ex{}{|\langle M(D) , X_0 \rangle - Y_0|} - \inf_{\theta \in \Reals^d} \ex{}{|\langle \theta , X_0 \rangle - Y_0|} \ge \ex{}{|\langle M(D) , X_0 \rangle - Y_0|} - \ex{}{|\langle P , X_0 \rangle - Y_0|} \ge \frac{\alpha}{d} \ex{}{\|\theta-P\|_1} - \frac{2\alpha\beta}{1+2\beta}.\]
Since $\theta,P \in [0,1]^d$, we also have $\|\theta-P\|_2^2 \le \|\theta-P\|_1$.
 Now that we have established these inequalities the proof more closely resembles a standard one-way marginals lower bound.

We invoke a form of the fingerprinting lemma:
\begin{lem}[{\cite[Lem.~10]{SteinkeU17}}]
Let $f: \{0,1\}^n \to \Reals$. Let $P \leftarrow \mathsf{Beta}(\beta,\beta)$ and, conditioned on $P$, let $Z_1, \cdots, Z_n \in \{0,1\}$ be independent with expectation $P$. Then 
\begin{align*}
    \ex{}{f(Z) \cdot \sum_i^n (Z_i-P)} &= 2\beta \cdot \ex{}{f(Z) \cdot \left(P-\frac12\right)} \\
    &= \beta \cdot \ex{}{ \left(f(Z) - \frac12 \right)^2 + \left( P - \frac12 \right)^2 - \left(f(Z) - P\right)^2}.
\end{align*}
\end{lem}
We apply this lemma to each coordinate $j$ with $\theta_j \mapsto f(Z)$ and $Z_j \mapsto Z$ and $P_j \mapsto P$. Summing over the $d$ coordinates gives
\begin{align*}
    \sum_i^n \ex{}{\langle \theta , Z_i-P \rangle} &= \beta \cdot \ex{}{\left\|\theta - \left(\frac12,\cdots,\frac12\right)\right\|_2^2 + \left\|P - \left(\frac12,\cdots,\frac12\right)\right\|_2^2 - \|\theta-P\|_2^2 }\\
    &\ge \beta \cdot \left( 0 + \sum_j^d \var{}{P_j} - \|\theta-P\|_2^2\right)\\
    &= \beta \cdot \left( \frac{d}{4(1+2\beta)} - \|\theta-P\|_2^2\right).%
\end{align*}

For the next part of the proof we think of $P$ has fixed, so that $(X_0,Y_0,Z_0), (X_1,Y_1,Z_2), \cdots, (X_n,Y_n,Z_n)$ are independent. We will also fix an arbitrary $i \in [n]$.

Since the algorithm $M$ only ``sees'' $Z_i$ via $Y_i=\langle Z_i , X_i \rangle$ (and $Z_i$ and $X_i$ are independent), $\theta$ is independent from $Z_{i,j}$ whenever $X_{i,j}=0$. For $i \in \{0\}\cup [n]$, let $J_i \in [d] \cup \{\bot\}$ denote the index of the coordinate such that $X_{i,J_i}=1$ and $X_{i,j}=0$ for all $j \ne J_i$ and $J_i=\bot$ if $X_i=\mathbf{0}$. Thus $\ex{}{\langle \theta, Z_i-P \rangle} = \ex{}{\theta_{J_i}\cdot(Z_{i,J_i}-P_{J_i})}$, where we define $\theta_\bot=P_\bot=Z_{i,\bot}=0$. Note that $J_i$ is independent from $Z_i$ and $P$.

Now it's time to use differential privacy.  Since $\theta$ is an $(\varepsilon,\delta)$-differentially private function of $D = ((X_1,Y_1),\cdots,(X_n,Y_n))$, we have \[\theta_{J_i}\cdot(Z_{i,J_i}-P_{J_i}) \approx_{\varepsilon,\delta} \theta_{J_0}\cdot(Z_{0,J_0}-P_{J_0}).\] That is, the distributions must be similar, since it is equivalent to changing only one sample. Namely, this is equivalent to replacing $(X_i,Y_i)$ with $(X_0,Y_0)$ in $D$. The latter distribution is easy for us to reason about because $\theta$ and $(Z_0,J_0)$ are independent (conditioned on $P$).

We will use the following lemma relating similarity of distributions to expectations.

\begin{lem}[{\cite[Lem.~A.1]{FeldmanS17}}]\label{lem:dp-ex}
 Let $U$ and $V$ be random variables supported on $[\mu-\Delta,\mu+\Delta]$. Suppose $U \approx_{\varepsilon,\delta} V$ -- that is, $\pr{}{U \in E} \le e^\varepsilon \pr{}{V \in E} + \delta$ and $\pr{}{V \in E} \le e^\varepsilon \pr{}{U \in E} + \delta$ for all events $E$. Then \[\left| \ex{}{U} - \ex{}{V} \right| \le \left( e^\varepsilon - 1 \right) \ex{}{|U-\mu|} + 2\delta\Delta. \]
\end{lem}

Thus, we have 
\begin{align*}
    \ex{}{\theta_{J_i}\cdot(Z_{i,J_i}-P_{J_i})} &\le \ex{}{\theta_{J_0}\cdot(Z_{0,J_0}-P_{J_0})} + (e^\varepsilon-1) \ex{}{|\theta_{J_0}\cdot(Z_{0,J_0}-P_{J_0})|} + 2\delta\\
    &= 0 + (e^\varepsilon-1) \ex{}{\theta_{J_0}\cdot|Z_{0,J_0}-P_{J_0}|} + 2\delta\\
    &\le (e^\varepsilon-1) \ex{}{|Z_{0,J_0}-P_{J_0}|} + 2\delta\\
    &= (e^\varepsilon-1) \ex{}{2P_{J_0}(1-P_{J_0})} + 2\delta\\
    &= (e^\varepsilon-1) \frac{\alpha\beta}{1+2\beta} + 2\delta.
\end{align*}
Combining inequalities and summing over $i \in [n]$ yields 
\[n \cdot \left( (e^\varepsilon-1) \frac{\alpha\beta}{1+2\beta} + 2\delta \right) \ge \sum_i^n \ex{}{\langle \theta , Z_i-P \rangle} \ge \beta \cdot \left( \frac{d}{4(1+2\beta)} - \ex{}{\|\theta-P\|_2^2} \right),\]
which rearranges to
\[ \ex{}{\|\theta-P\|_2^2} \ge \frac{d - 4n(e^\varepsilon-1)\alpha}{4+8\beta}  - \frac{2\delta n}{\beta}.\]
Now
\begin{align*}
\ex{}{|\langle M(D) , X_0 \rangle - Y_0|} -\ex{}{|\langle P , X_0 \rangle - Y_0|} &\ge \frac{\alpha}{d} \ex{}{\|\theta-P\|_1} - \frac{2\alpha\beta}{1+2\beta}\\
    &\ge \frac{\alpha}{d} \ex{}{\|\theta-P\|_2^2} - \frac{2\alpha\beta}{1+2\beta} \\
    &\ge  \alpha \frac{1 - 4n(e^\varepsilon-1)\alpha/d}{4+8\beta}  - \frac{2\alpha\delta n}{\beta d} - \frac{2\alpha\beta}{1+2\beta} \\
    &=  \alpha \frac{1 - 8\beta - 4n(e^\varepsilon-1)\alpha/d}{4+8\beta}  - \frac{2\delta \alpha n}{\beta d}.
\end{align*}
We set $\beta=1/80$ to obtain \[ \ex{}{|\langle M(D) , X_0 \rangle - Y_0|}  - \ex{}{|\langle P , X_0 \rangle - Y_0|}  \ge  \alpha \frac{0.9 - 4n(e^\varepsilon-1)/d}{4.1}  - \frac{160\delta \alpha n}{d}.\]
Finally, we make use of the assumptions $n \le \frac{d}{10(e^\varepsilon-1)\alpha}$ and $\delta \le \frac{e^\varepsilon-1}{1000}$ to conclude \[ \ex{}{|\langle M(D) , X_0 \rangle - Y_0|}  -  \ex{}{|\langle P , X_0 \rangle - Y_0|}  \ge  \alpha \frac{0.9 - 0.4}{4.1}  - \frac{16\delta \alpha}{e^\varepsilon-1} \ge \frac{\alpha}{10}.\]
\end{proof}

Theorem \ref{thm:lb-core} contains all the technical ingredients of our lower bound. The only thing that is missing is that it is in terms of population risk, rather than empirical risk. We address this next.

\begin{cor}
Let $M : (\bfR^d \times \bfR)^n \to \bfR^d$ be a $(\varepsilon,\delta)$-differentially private algorithm with $\varepsilon \le \frac{1}{21}$ and $\delta \le \frac{\varepsilon}{1000}$. Then there exists $D = ((\bfx_1,y_1),\cdots,(\bfx_n,y_n)) \in (\bfR^d \times \bfR)^n$ with $\|\bfx_i\|_2\le1$ and $y_i \in \{0,1\}$ for all $i \in [n]$ such that \[\ex{}{\frac1n \sum_{i=1}^n |\langle M(D), \bfx_i \rangle - y_i| - |\langle \theta^* , \bfx_i \rangle - y_i|} \ge  \min \left\{ \frac{1}{20} , \frac{d}{210 \varepsilon n}\right\} - 2\delta,\] where $\theta^* := \argmin_{\theta \in \Reals^d} \sum_{i=1}^n |\langle \theta^* , \bfx_i \rangle - y_i|$ satisfies $\theta^* \in [0,1]^d$ (assuming we break ties in the argmin towards lower-norm vectors).
\end{cor}
To obtain Theorem \ref{thm:ERM-lb} from this corollary we simply need to pad the features in the dataset $D$ with zeros to attain dimensionality $p \ge d$. Note that $\rank(\sum_i \bfx_i \bfx_i^T) \le d$ and $\|\theta^*\|_2 \le \sqrt{d}$ and $\|\bfx_i\|_2\le 1$ even after this padding.
\begin{proof}
We set $\alpha = \min \left\{ 1 , \frac{d}{10(e^\varepsilon-1)n}\right\}$ and invoke Theorem \ref{thm:lb-core} to pick some $p \in [0,1]^d$ such that the following holds. Let the distribution $\mathcal{D}_{p,\alpha}$ on $\{0,1\}^d \times \{0,1\}$ be as in Theorem \ref{thm:lb-core}. Let $(X_0,Y_0), (X_1,Y_1),\cdots,(X_n,Y_n)$ be independent draws from $\mathcal{D}_{p,\alpha}$ and set $D=((X_1,Y_1),\cdots,(X_n,Y_n))$. Then \[\ex{}{|\langle M(D), X_0 \rangle - Y_0| - |\langle p , X_0 \rangle - Y_0|} \ge \frac{\alpha}{10} = \min \left\{ \frac{1}{10} , \frac{d}{100(e^\varepsilon-1)n}\right\}.\]
We use Lemma \ref{lem:dp-ex} and the differential privacy guarantee to relate the population value above with the emprical value: For all $i \in [n]$, we have
\begin{align*}
    &\left| \ex{}{|\langle M(D), X_i \rangle - Y_i| - |\langle p , X_i \rangle - Y_i|} - \ex{}{|\langle M(D), X_0 \rangle - Y_0| - |\langle p , X_0 \rangle - Y_0|}  \right| \\
    &~~\le (e^\varepsilon-1) \ex{}{\left| |\langle M(D), X_0 \rangle - Y_0| - |\langle p , X_0 \rangle - Y_0| \right|} + 2\delta\\
    &~~\le (e^\varepsilon-1) \pr{}{X_0 \ne \mathbf{0}} \cdot 1 + 2\delta\\
    &~~= (e^\varepsilon-1) \cdot \alpha + 2\delta.
\end{align*}
Averaging over $i \in [n]$ gives
\[\frac1n \sum_{i=1}^n \ex{}{|\langle M(D), X_i \rangle - Y_i| - |\langle p , X_i \rangle - Y_i|} \ge \frac{\alpha}{10} - (e^\varepsilon-1) \cdot \alpha - 2\delta.\]
Since $\varepsilon \le \frac{1}{21}$, we have $e^\varepsilon-1 \le \frac{21}{20} \varepsilon \le \frac{1}{20}$ and
\[\ex{}{\frac1n \sum_{i=1}^n |\langle M(D), X_i \rangle - Y_i| - \inf_{\theta \in \bfR^d} |\langle \theta , X_i \rangle - Y_i|} \ge \frac1n \sum_{i=1}^n \ex{}{|\langle M(D), X_i \rangle - Y_i| - |\langle p , X_i \rangle - Y_i|} \ge  \frac{\alpha}{20} - 2\delta.\]
Here $D$ is still a random dataset. Depending on $M$, we can pick some fixed dataset from the support such that the inequality holds. The only remaining randomness is that of the algorithm $M$.

Finally, we have \[\theta^* = \argmin_{\theta \in \Reals^d} \sum_{i=1}^n |\langle \theta^* , X_i \rangle - Y_i| = \argmin_{\theta \in \Reals^d} \sum_{i=1}^n |\langle \theta^* - Z_i , X_i \rangle |.\]
Since $Z_i \in \{0,1\}^d$ for all $i \in [n]$, we conclude by convexity that $\theta^* \in [0,1]^d$. In fact, it can be shown that $\theta^* \in \{0,1\}^d$ is the coordinate-wise majority of the $Z_i$s.
\end{proof}

\section{Proofs for Section~\ref{sec:privNonConvex}}
\label{app:nonconvex}
\begin{proof}[Proof of Theorem~\ref{thm:util31}]
Recall that $M$ is the projector to the eigenspace of the matrix $\sum\limits_{i=1}^n \bfx_i\bfx_i^T$, and $\snorm{\cdot}{M}$ being the corresponding seminorm. Let $\theta_1,\ldots,\theta_T$ be the sequence of models generated in Line 4 of Algorithm \ref{Alg:Frag}, and let the constraint set $\calC=\bfR^p$. 
Also, let $\boldb_t$ be the Gaussian noise added in the $t$-th iteration.
By the smoothness property of $\ell(z;\cdot)$, we have the following:
\begin{align}
    \calL(\theta_{t+1};D)&\leq \calL(\theta_t;D)+\ip{\nabla \calL(\theta_t;D)}{\theta_{t+1}-\theta_t}_M+\frac{\beta}{2}\snorm{\theta_{t+1}-\theta_t}{M}^2\nonumber\\
    &=\calL(\theta_t;D)-\frac{1}{\beta}\ip{\nabla\calL(\theta_t;D)}{\nabla\calL(\theta_t;D)+\boldb_t}_M+\frac{1}{2\beta}\snorm{\nabla\calL(\theta_t;D)+\boldb_t}{M}^2\nonumber\\
    &=\calL(\theta_t;D)-\frac{1}{2\beta}\snorm{\nabla\calL(\theta_t;D)}{M}^2+\frac{\snorm{\boldb_t}{M}^2}{2\beta}\nonumber\\
    \Leftrightarrow\quad \snorm{\nabla\calL(\theta_t;D)}{M}^2&\leq 2\beta\left(\calL(\theta_t;D)-\calL(\theta_{t+1};D)\right)+\snorm{\boldb_t}{M}^2.\label{eq:asc4}
\end{align}
Therefore, averaging over all the $t\in\{0,\dots,T-1\}$, we have the following:
\begin{align}
\frac{1}{T}\sum\limits_{t=0}^{T-1} \snorm{\nabla\calL(\theta_t;D)}{M}^2 &\leq \frac{2\beta}{T}\left(\calL(\boldsymbol{0};D)-\calL(\theta_{T};D)\right)+\frac{1}{T}\sum\limits_{t=0}^{T}\snorm{\boldb_t}{M}^2\nonumber\\&\leq \frac{2\beta}{T}\left(\calL(\boldsymbol{0};D)-\calL(\theta^*
;D)\right)+\frac{1}{T}\sum\limits_{t=1}^{T}\snorm{\boldb_t}{M}^2.\label{eq:asc5}
\end{align}
Using standard Gaussian concentration, w.p. at least $1-\gamma$ over the randomness of $\{\boldb_1,\ldots,\boldb_T\}$ in \eqref{eq:asc5}, we have the following.
\begin{equation}
\frac{1}{T}\sum\limits_{t=0}^{T-1} \snorm{\nabla\calL(\theta_t;D)}{M}^2\leq \frac{2\beta}{T}\left(\calL(\boldsymbol{0};D)-\calL(\theta^*;D)\right)+\frac{8L^2\rank(M)\cdot\log(1/\delta)\log(T/\gamma)}{n^2\epsilon^2} \label{eq:asc6}
\end{equation}
By an averaging argument, we know there exists $\hat{t}\in\{0,\dots,T-1\}$ s.t.  
$$\snorm{\nabla\calL(\theta_{\hat{t}};D)}{M}^2\leq \frac{2\beta}{T}\left(\calL(\boldsymbol{0};D)-\calL(\theta^*;D)\right)+\frac{8L^2\rank(M)\cdot\log(1/\delta)\log(T/\gamma)}{n^2\epsilon^2}.$$
As long as $T\geq \frac{\beta n^2\epsilon^2\cdot\calL(\mathbf{0};D)}{2L^2\log(1/\delta)}$, we have $\snorm{\nabla\calL(\theta_{\hat{t}};D)}{M}\leq \frac{4L\sqrt{\rank(M)\cdot\log(1/\delta)\log(T/\gamma)}}{\epsilon n}$. Now, notice that the $\ell_2$-sensitivity \citep{dwork2014algorithmic} of $\snorm{\nabla\calL(\theta_{\hat{t}};D)}{M}$ is at most $\frac{2L}{n}$. Therefore, releasing $t_{\sf priv}\leftarrow \argmin\limits_{t\in\{0,\dots,T-1\}}\snorm{\nabla\calL(\theta_t;D)}{M}+{\sf Lap}\left(\frac{4L}{n}\right)$ conditioned on $\theta_0,\ldots,\theta_T$ satisfies $\epsilon$-differential privacy (by the analysis of the report-noisy-max algorithm \citep{dwork2014algorithmic}). Therefore, the whole algorithm is $(2\epsilon,\delta)$-differentially private.

As for utility, we have w.p. at least $1-\gamma$,
$$\snorm{\nabla\calL(\theta_{t_{\sf priv}};D)}{2}=\snorm{\nabla\calL(\theta_{t_{\sf priv}};D)}{M}=O\left(\frac{L\sqrt{\rank(M)\cdot\log(1/\delta)\log(T/\gamma)}}{\epsilon n}\right).$$
Here, we have used the standard concentration property of Laplace random variable. This completes the proof.
\end{proof}

\section{Proofs and More Details for Section \ref{sec:clipping}}
\label{app:biasControl}

\subsection{Generic Tool for Understanding Clipping}
\label{app:undclip}
We first define some notations.
\begin{itemize}
\item For any vector $v$ and positive scalar $I$, let $\trunc{I}{v}$ denote $\min\left\{\frac{I}{\|v\|_2}, 1\right\} \cdot x$, i.e., $x$ projected onto the $\ell_2$-ball of radius $I$. 
If $v$ is a scalar, then $\trunc{I}{v} = \max\{\min\{v,I\},-I\}$.
Also, for scalar, we use $\trunc{I^+}{v}$ to denote $\min\{v,I\}$, and $\trunc{I^-}{v}$ to denote $\max\{v,-I\}$.
\item For a set $S$ of scalar or vector, let $\trunc{I}{S}$ denote $\bbracket{\trunc{I}{v}:v\in S}$. 
For a set $S$ of scalar, let $\trunc{I^+}{S} = \bbracket{\trunc{I^+}{v}:v\in S}$ and $\trunc{I^-}{S} = \bbracket{\trunc{I^-}{v}:v\in S}$.
\item For a set $S$ of scalar, we write $S > I$ if $\forall u \in S$, $u > I$; similar for $<$, $\geq$ and $\leq$.
\end{itemize}

\begin{proof}[Proof of Lemma~\ref{lem:truncation}]
We consider any fixed $\boldx$, and for simplicity we use $g$ to denote $g_{\boldx}$.
Let $\cn' = \cn / \|\boldx\|_2$.
We first show $g$ is convex and $\partial g(y) = \trunc{\cn'}{\partial f(y)}$ for $y \in \bfR$ using the following claims. And then apply that to $\ell_f$ and $\ell_g$ to prove the theorem.

\begin{claim}\label{claim:basic}
The following holds.
\begin{enumerate}
\item 
If $y_1 \in \bfR$, then $-\cn' \in \partial f(y_1)$, and thus $\partial f(y) \geq -\cn'$ for all $y > y_1$, $\partial f(y) \leq -\cn'$ for all $y < y_1$.
If $y_2 \in \bfR$, then $\cn' \in \partial f(y_2)$, and thus $\partial f(y) \leq \cn'$ for all $y < y_2$, $\partial f(y) \geq \cn'$ for all $y > y_2$.
\item
If $y_1 = -\infty$, then $\partial f(y) \geq -\cn'$ for all $y\in \bfR$.
If $y_2 = \infty$, then $\partial f(y) \leq \cn'$ for all $y\in \bfR$.
\item
If $y_1 = \infty$, then $\partial f(y) < -\cn'$ for all $y\in \bfR$, and thus $y_2 = \infty$.
If $y_2 = -\infty$, then $\partial f(y) > \cn'$ for all $y\in \bfR$, and thus $y_1 = -\infty$.
\end{enumerate}
\end{claim}
\begin{proof}
We consider the three cases separately.
\begin{enumerate}
\item 
By definition of $y_2$ and monotonicity of subdifferential, for $y > y_2$, we have $\partial f(y) > \cn'$, and for $y < y_2$, $\min \partial f(y) \leq \cn'$ (as subdifferential is closed).

We can find a sequence $y^{(k)} \to y_0^+$, and a sequence $g^{(k)}$ with $g^{(k)} \in \partial f(y^{(k)})$. As subdifferential is monotone, we know $g^{(k)}$ is decreasing. Since $g^{(k)}$ is lower bounded by $I$, the sequence $g^{(k)}$ converges to a value $\geq \cn'$. 
Similarly, we can find a sequence $y^{(k')} \to y_2^-$, and a sequence $g^{(k')}$ with $g^{(k')} = \min \partial f(y^{(k')})$. Since $g^{(k')}$ is increasing and upper bounded by $\cn'$, the sequence $g^{(k')}$ converges to a value $\leq \cn'$. 

Recall that subdifferential is continuous. So both $\lim_{k\to\infty}g^{(k)} \geq \cn'$ and $\lim_{k'\to\infty} g^{(k')} \leq \cn'$ are contained in $\partial f(y_0)$, and we have $\cn' \in \partial f(y_2)$ by convexity of subdifferential. 
Then by monotonicity, for any $y > y_2$, we have $\partial f(y) \geq \cn'$; for any $y < y_2$, we have $\partial f(y) \leq \cn'$.
Similar argument can be applied to $y_1$.
\item
If $Y_1$ is non-empty and $y_1 = \infty$, by monotonicity of subdifferential, we have $\partial f(y) < -\cn'$ for all $y\in \bfR$. Therefore, $Y_2 = \emptyset$ and we have $y_2 = \infty$. Similar holds for $y_2$.
\item
If $Y_1$ is empty and $y_1 = -\infty$, then for any $y$, $\max \partial f(y) \geq -\cn'$ (as subdifferential is closed). By monotonicity of subdifferential, $\partial f(y) \geq -\cn'$ for all $y\in \bfR$. 
\end{enumerate}
\end{proof}

By monotonicity of subdifferential and the definition of $y_1$ and $y_2$, we know that $y_1 \leq y_2$ always holds and thus $g$ is well-defined. 
Let $\calC = [y_1, y_2]\cap \bfR$.

\begin{claim}
We have that $g$ is a convex function when $y_1\neq\infty$ and $y_2\neq-\infty$.
\end{claim}
\begin{proof}
By Claim~\ref{claim:basic}, for any $y \in (y_1, y_2)$, $-\cn' \leq \partial f(y) \leq \cn'$. Therefore, $f$ is $\cn'$-Lipschitz on $[y_1, y_2]\cap \bfR$. It is obviously also convex on this set.
Consider $f$ restricted to $\calC$. According to~\cite[Lemma 6.3]{BST14}, the Lipschitz extension of this function, $\hat{g}:\bfR\to\bfR$ with $\hat{g}(y) = \min_{y'\in\calC} \bbracket{f(y') + \cn'\abs{y-y'}}$, is also convex and $\cn'$-Lipschitz.

Then we show $g = \hat{g}$.
For any $y \in \calC$, we have $f(y) \leq f(y') + \cn'\abs{y-y'}$ by Lipschizness; so $\hat{g}(y)=f(y)=g(y)$ on $\calC$.
If $y_1\neq -\infty$, for any $y < y_1$ and any $y' \in \calC$, we have $f(y_1) - f(y') \leq \cn'(y'-y_1)$ by Lipschitzness and $y' \geq y_1$. This translates to $f(y_1) - \cn'(y-y_1) \leq f(y') - \cn'(y-y')$, and thus $\hat{g}(y) = f(y_1) - \cn'(y-y_1) = g(y)$ for $y < y_1$.
Similar holds for $y > y_2$ when $y_2 \neq \infty$.
\end{proof}

\begin{claim}
We have $\partial g(y) = \trunc{\cn'}{\partial f(y)}$ for $y \in \bfR$.
\end{claim}
\begin{proof}
If $y_1 = \infty$, then $g$ is a linear function with coefficient $-\cn'$, and is obviously convex. By Claim~\ref{claim:basic}, we have $\partial f(y) < -\cn'$ on $\bfR$, and thus $\trunc{\cn'}{\partial f(y)} = \{-\cn'\} = \partial g(y)$ for all $y \in \bfR$. Similar holds for $y_2 = -\infty$.

Now we consider the case where $y_1\neq\infty$ and $y_2\neq-\infty$.
\begin{itemize}
\item 
On $(-\infty, y_1)$, $g$ is linear and thus differentiable, so $\partial g(y) = \{-\cn'\}$. Also, we know from Claim~\ref{claim:basic} that $\partial f(y) \leq -\cn'$ for $y\in (-\infty, y_1)$; so $\trunc{\cn'}{\partial f(y)} = \{-\cn'\} = \partial g(y)$. Similar holds for $(y_2, \infty)$.

\item
On $(y_1, y_2)$, we have $g = f$ and both are convex.
Any convex function $f$ on an open subset of $\bfR$ is semi-differentiable and the subdifferential at point $y$ is of the form $[\partial f_-(y), \partial f_+(y)]$ where $\partial f_-(y)$ is the left derivative and $\partial f_+(y)$ is the right derivative.
Therefore, as the left and right derivative of $f$ and $g$ are the same in $(y_1, y_2)$, we have $\partial g(y) = \partial f(y)$. By Claim~\ref{claim:basic}, $-\cn' \leq \partial f(y) \leq \cn'$ on this range, we have $\partial g(y) = \trunc{\cn'}{\partial f(y)}$. 

\item
At $y_1$ (if finite), the left derivative is $\partial g_-(y) = -\cn'$, and the right derivative $\partial g_+(y)$ is $\partial f_+(y)$ if $y_2 > y_1$ and is $\cn'$ if $y_2 = y_1$. 
For $y_2 > y_1$, as $-\cn' \in \partial f(y_1)$, we have $\partial g(y) = \trunc{\cn'}{\partial f(y)}$.
For $y_2 = y_1$, we have $[-\cn', \cn'] \subseteq \partial f(y_1)$ and thus $\partial g(y) = \trunc{\cn'}{\partial f(y)}$.
Similar holds for $y_2$.
\end{itemize}
\end{proof}

Then we consider $\ell_g$.
For a set $U$ of scalar, we use $U\cdot \boldx$ to denote $\bbracket{u\boldx:u\in U}$.
We have $\trunc{\cn}{U\cdot \boldx} = \bbracket{\trunc{\cn}{u \boldx} : u\in U}$ $= \bbracket{\min\bbracket{\frac{\cn}{\snorm{\boldx}{2} u}, 1} u \boldx: u\in U} = \trunc{\cn/\snorm{\boldx}{2}}{U} \cdot \boldx$.
Recall $\cn' = \cn / \snorm{\boldx}{2}$.
Therefore, $\partial_{\theta} \ell_g(\theta;\boldx) = \partial g(\ip{\theta}{\boldx}) \cdot \boldx = \trunc{\cn/\snorm{\boldx}{2}}{\partial f(\ip{\theta}{\boldx})} \cdot \boldx = \trunc{\cn}{\partial f(\ip{\theta}{\boldx}) \cdot \boldx} = \trunc{\cn}{\partial_{\theta} \ell_f(\ip{\theta}{\boldx})}$, which completes the proof.
\end{proof}

\subsection{Lower Bound on Bias for Binary Logistic Regression}
\label{sec:lower bound logistic regression}

Using Lemma \ref{lem:huberized}, in Theorem \ref{thm:LR_lower}, we show that running \dpgdclip\ with aggressive clipping can result in a constant excess empirical risk for logistic regression, in contrast to the best achievable excess empirical risk of $\bigO{1/n}$.  

For a dataset $D = \{(x_1, y_1),\dots, (x_n, y_n)\}$ where $x_i\in \bfR^p$ is the feature and $y_i \in \{+1, -1\}$ is the label, and for a convex set $\calC$, logistic regression is defined as solving for
$\theta^* \vcentcolon= \argmin_{\theta \in \calC} \calL(\theta; D)$
where
$\calL(\theta; D) = \frac{1}{n}\sum_{i=1}^n \ell(\theta; (x_i,y_i))$ with $\ell(\theta; (x,y)) = \log \bracket{1 + e^{-y\ip{\theta}{x}}}$.

\begin{thm}\label{thm:LR_lower}
Consider the objective function $\calL(\theta, D)$ for logistic regression as defined above.
Let $\privT$ be the output of \dpgdclip\ on $\calL(\theta, D)$ with clipping norm $\cn$.
For any $\cn < 1/4$, there exists a positive integer $n_0(\cn)$ such that for any $n \geq n_0(\cn)$, there exists a dataset $D = \{(x_i,y_i)\}_{i=1}^n$ with $x_i\in \{x\in\bfR^p:\|x\|_2\leq 1\}$ and $y_i \in \{+1, -1\}$, such that
\begin{align*}
\expt{}{\calL(\privT; D)} - \min_{\theta\in \bfR^p}\calL(\theta; D) = \bigOmega{\log(1/\cn)}.
\end{align*}
\end{thm}
\mypar{Note} The lower bound construction does not require constraining $\calC$.
With $\calC$ being the whole space $\bfR^p$, the optimization problem considered here is unconstrained. 
Also, notice that $\cn < 1/4$ is not a strong requirement, as for any $(x_i,y_i)$, the gradient of logistic loss is upper bounded by $\|x_i\|_2 \leq 1$. So, $\cn = 1$ is already equivalent to no clipping.

It is obvious that if we set the clipping norm $\cn$ to be higher than the upper bound of the gradient norm, 
which exists in both cases,
then we can still get $\tilde O(1/n)$ excess empirical risk. 
Therefore, we can conclude that picking a proper $\cn$ is critical in convex optimization problems.

\begin{proof}[Proof of Theorem~\ref{thm:LR_lower}]
Since $\ell(\theta; (x,y))$ is convex in $\ip{\theta}{yx}$, as have been shown Lemma~\ref{lem:truncation} (with $\boldx$ there being $yx$), 
for any $(x,y)$, there exists another function $\ell_g(\theta; (x, y))$ that is convex in $\ip{\theta}{yx}$ and $\nabla_{\theta} \ell_g(\theta; (x, y)) = \trunc{\cn}{\nabla_{\theta} \ell(\theta; (x, y))}$ for any $\theta$.
Let $\calLH(\theta; D) = \frac{1}{n} \sum_{i=1}^n \ell_g(\theta; (x_i,y_i))$, which is also convex. 
For some convex set $\calC$, let $\thetaH \vcentcolon= \argmin_{\theta\in \calC} \calLH(\theta; D)$ and $\theta^* \vcentcolon= \argmin_{\theta\in\calC} \calL(\theta; D)$. Let $\privT$ be the output of DP-GD on objective function $\calL$.

To show a lower bound on $\expt{}{\calL(\privT; D)} - \calL(\theta^*; D)$, we would first show a lower bound on $\|\theta^*-\thetaH\|_2$ and an upper bound on $\|\privT-\thetaH\|_2$, which together will give a lower bound on $\|\theta^*-\thetaH\|_2$. Then, using strong convexity property of $\calL$, we translate that to lower bound on $\calL(\privT;D)-\calL(\theta^*;D)$.

It is enough to prove the result for dimension $p=1$, as we can always set the other $p-1$ dimensions to be $0$.
Let $D$ be $\{(1/2, +1)\}^{2n} \cup \{(1, -1)\}^n$ and $\calC = \bfR$.

We have
\begin{align*}
& \ell(\theta; (1/2,+1)) = \log \bracket{1 + e^{-\theta/2}},  & \text{and} \quad
 & \ell(\theta; (1,-1)) = \log \bracket{1 + e^{\theta}} \\
\Rightarrow & \nabla_{\theta} \ell(\theta; (1/2,+1)) = -\frac{1/2}{1+e^{\theta/2}},  & \text{and} \quad
 & \nabla_{\theta} \ell(\theta; (1,-1)) = \frac{1}{1+e^{-\theta}} \\
\Rightarrow & \nabla^2_{\theta} \ell(\theta; (1/2,+1)) = \frac{1}{4} \frac{1}{(1+e^{\theta/2})(1+e^{-\theta/2})}, & \text{and} \quad
 & \nabla^2_{\theta} \ell(\theta; (1,-1)) = \frac{1}{(1+e^{\theta})(1+e^{-\theta})}
\end{align*}
Given $\cn$, we have
\begin{align*}
&\ell_g(\theta; (1/2,+1))
=
\begin{cases}
- \cn \theta + 2\cn\log\bracket{\frac{1}{2\cn}-1} + \log\frac{1}{1-2\cn} \quad &\text{for } \theta < 2\log\bracket{\frac{1}{2\cn}-1} \\
\log\bracket{1+e^{-\theta/2}}  \quad &\text{for } \theta \geq 2\log\bracket{\frac{1}{2\cn}-1} \\
\end{cases}\\
& \qquad \qquad \qquad \qquad \qquad \qquad \text{and} \\
&\ell_g(\theta; (1,-1))
=
\begin{cases}
\log\bracket{1+e^{\theta}}  \quad &\text{for } \theta \leq -\log\bracket{\frac{1}{\cn}-1} \\
\cn \theta + \cn\log\bracket{\frac{1}{\cn}-1} + \log\frac{1}{1-\cn} \quad &\text{for } \theta > -\log\bracket{\frac{1}{\cn}-1} \\
\end{cases}
\end{align*}
as the loss function with gradient being $\trunc{\cn}{\nabla_{\theta} \ell(\theta; (1/2,+1))}$ and $\trunc{\cn}{\nabla_{\theta} \ell(\theta; (1,-1))}$.

We have $\theta^* = 0$ as
\begin{align*}
\nabla_{\theta} \calL(\theta;D) = 0
& \Leftrightarrow 
2 \nabla_{\theta} \ell(\theta; (1/2,+1)) + \nabla_{\theta} \ell(\theta; (1,-1)) = 0 \\
& \Leftrightarrow 
-\frac{1}{1+e^{\theta/2}} + \frac{1}{1+e^{-\theta}} = 0
\Leftrightarrow
\theta = 0.
\end{align*}

As for $\thetaH$, we have
\begin{align*}
\nabla_\theta \calLH(\theta; D)
=
\begin{cases}
\frac{2}{3} \frac{-1/2}{1+e^{\theta/2}} + \frac{\cn}{3} &\text{\quad for } \theta \geq 2\log\bracket{\frac{1}{2\cn}-1} \\
-\frac{2\cn}{3} + \frac{\cn}{3} &\text{\quad for } \theta \in  \bracket{-\log\bracket{\frac{1}{\cn}-1}, 2\log\bracket{\frac{1}{2\cn}-1}} \\
-\frac{2\cn}{3} + \frac{1}{3} {\frac{1}{1+e^{-\theta}}} &\text{\quad for } \theta \leq -\log\bracket{\frac{1}{\cn}-1}
\end{cases},
\end{align*}
which is equal to $0$ at $2 \log \bracket{\frac{1}{\cn} - 1}$. So we have $\thetaH = 2 \log \bracket{\frac{1}{\cn} - 1}$.

So we have $\|\theta^* - \thetaH\|_2 = 2 \log \bracket{\frac{1}{\cn} - 1}$.

Now we bound $\|\thetaH - \privT\|_2$.
As for each $x_i$ in $D$, it is given that $\|x_i\|_2 \leq 1$, we have $B = 1$. 
As the initial guess $\theta_0$ is $0$, we have $\|\theta_0 - \theta^*\|_2 = 2\log\bracket{\frac{1}{\cn}-1}$.
From Theorem~\ref{thm:dpsgd_convergence}, with probability $\geq 1-\delta$,  
$\calLH(\privT; D) - \calLH(\thetaH; D) \leq C\cdot{\frac{2\log\bracket{1/\cn-1} \sqrt{p\log(1/\delta)\log(1/\beta)}}{n \epsilon}}$ for some positive constant $C$. 
We set $\frac{n_0}{3} = \frac{96C\log\bracket{1/\cn-1} \sqrt{p\log(1/\delta)\log(1/\beta)}}{\cn\epsilon} > \max \bracket{20\log\bracket{\frac{1}{\cn}-1}, 96} \cdot \frac{C \sqrt{p\log(1/\delta)\log(1/\beta)}}{\cn\epsilon}$.
As $n \geq \frac{n_0}{3} > \frac{20C\log\bracket{1/\cn-1} \sqrt{p\log(1/\delta)\log(1/\beta)}}{\cn\epsilon}$, we have $\calLH(\privT; D) - \calLH(\thetaH; D) < 0.1\cn$.
We now translate this to an upper bound on $\|\privT-\thetaH\|_2$ (with high probability).

Let $\theta_1 = 2\log\bracket{\frac{1}{2\cn}-1}$ and $\theta_2 = 2\log\bracket{\frac{2}{\cn}-1}$. We now show $\privT \in (\theta_1, \theta_2)$ with probability $\geq 1-\delta$.
We know that for $\theta \geq \theta_1$, for ${\sf Const} = \frac{1}{3}\bracket{\cn\log\bracket{\frac{1}{\cn}-1} + \log \frac{1}{1-\cn}}$,
\begin{align*}
\calLH(\theta; D)
= \frac{1}{3} \bracket{2\log(1+e^{-\theta/2}) + \cn\theta} + {\sf Const},
\end{align*}
and we thus have
\begin{align*}
\calLH(\thetaH; D) 
&= \frac{1}{3} \bracket{2\log\frac{1}{1-\cn} + 2\cn\log\bracket{\frac{1}{\cn}-1}} + {\sf Const}\\
\calLH(\theta_1; D) 
&= \frac{1}{3} \bracket{2\log\frac{1}{1-2\cn} + 2\cn\log\bracket{\frac{1}{2\cn}-1}} + {\sf Const} \\
\calLH(\theta_2; D) 
&= \frac{1}{3} \bracket{2\log\frac{2}{2-\cn} + 2\cn\log\bracket{\frac{2}{\cn}-1}} + {\sf Const}.
\end{align*}
So for $\cn < 1/4$,
\begin{align*}
\calLH(\theta_1; D) - \calLH(\thetaH; D) 
&= \frac{2}{3} \Bigg(\log\frac{1}{1-2\cn} + \cn\log\bracket{\frac{1}{2\cn}-1}  - \log\frac{1}{1-\cn} - \cn\log\bracket{\frac{1}{\cn}-1}\Bigg) \\
&= \frac{2}{3} \bracket{(1-\cn)\log\frac{1-\cn}{1-2\cn} - \cn\log 2}
\geq \frac{2}{3} \bracket{1-\log 2} \cn 
> 0.2\cn > 0.1\cn. \\
\calLH(\theta_2; D) - \calLH(\thetaH; D) 
&= \frac{2}{3} \Bigg(\log\frac{2}{2-\cn} + \cn\log\bracket{\frac{2}{\cn}-1} - \log\frac{1}{1-\cn} - \cn\log\bracket{\frac{1}{\cn}-1}\Bigg) \\
&= \frac{2}{3} \Bigg(\log\frac{2}{2-\cn} + \cn\log\bracket{\frac{2}{\cn}-1} - \log\frac{1}{1-\cn} - \cn\log\bracket{\frac{1}{\cn}-1}\Bigg) \\
&\geq \frac{2}{3} \bracket{\log(2)-\frac{1}{2}} \cn
> 0.1\cn.
\end{align*}

Notice that $\calLH$ is convex, which means the derivative is monotone and the function is decreasing for $\theta < \thetaH$ and increasing for $\theta > \thetaH$. As $\theta_1 < \privT < \theta_2$, if $\privT \leq \theta_1$ or $\privT \geq \theta_2$, then $\calLH(\privT; D) - \calLH(\thetaH; D) \geq 0.1\cn$, which contradicts to the fact that $\calLH(\privT; D) - \calLH(\thetaH; D) < 0.1\cn$. Therefore, we can conclude that $\privT \in (\theta_1, \theta_2)$ with probability $\geq 1-\delta$.

For $\theta \in \bracket{\theta_1, \theta_2}$, the 2nd order derivative of $\calLH$ is $\frac{1}{6} \frac{1}{(1+e^{\theta/2}) (1+e^{-\theta/2})}
\geq \frac{1}{12} \frac{1}{1+e^{\theta/2}}$, which is decreasing and therefore $\geq \frac{\cn}{24}$. This means $\calLH$ is $\frac{\cn}{24}$-strongly convex for $\theta$ in this range.
Therefore, the bound on the difference of the loss translates to a bound on the $\ell_2$ distance and we have $\|\privT - \thetaH\|_2^2 \leq C \cdot {\frac{48 \log\bracket{2/\cn-1} \sqrt{p\log(1/\delta)\log(1/\beta)}}{\cn n \epsilon}}$.
We then have 
\begin{align*}
\|\privT - \theta^*\|_2
&\geq \|\theta^* - \thetaH\|_2 - \|\privT - \thetaH\|_2 \\
&\geq 2 \log \bracket{\frac{1}{\cn} - 1} - {\sqrt{\frac{48C\log(2/\cn-1)\sqrt{p\log(1/\delta)\log(1/\beta)}}{\cn n \epsilon}}} \\
&\geq \log \bracket{\frac{1}{\cn} - 1}
\end{align*}
where the last inequality follows as for $n > \frac{96 C}{\cn} \frac{\sqrt{p\log(1/\delta)\log(1/\beta)}}{\epsilon}$, we have
$$ \sqrt{\frac{48C\log(2/\cn-1)\sqrt{p\log(1/\delta)\log(1/\beta)}}{\cn n \epsilon}} \leq \log \bracket{\frac{1}{\cn} - 1} \text{ for any } \cn < 1/4.$$

Let $\theta_3 = -\log\bracket{\frac{1}{\cn}-1}$ and $\theta_4 = \log\bracket{\frac{1}{\cn}-1}$. 
As $\theta^* = 0$, the above inequality implies $\privT \in (-\infty, \theta_3]\cup[\theta_4, \infty)$. 
Similarly, as $\calL$ is convex with minimizer $\theta^*$, we know $\calL(\privT;D) \geq \min\bracket{\calL(\theta_3;D), \calL(\theta_4;D)}$.

Since 
\begin{align*}
\calL(\theta;D) 
= \frac{1}{3}\bracket{2\log\bracket{1+e^{-\theta/2}} + \log\bracket{1+e^{\theta}}}
= \frac{1}{3} \log\bracket{\power{1+e^{\theta/2}}{2} + \power{1+e^{-\theta/2}}{2}}
\end{align*}
is an even function, we have
\begin{align*}
\calL(\theta_3;D) = \calL(\theta_4;D)
= \frac{1}{3} \log \bracket{\power{1+\sqrt{\frac{1-\cn}{\cn}}}{2} + \power{1+\sqrt{\frac{\cn}{1-\cn}}}{2}}
= \frac{2}{3} \log\bracket{\frac{1}{\sqrt{\cn}}+\frac{1}{\sqrt{1-\cn}}}.
\end{align*}

As $\calL(\theta^*;D) = \log 2$, we have, for $\cn < 1/4$,
\begin{align*}
\calL(\privT;D) - \calL(\theta^*;D)
&\geq \frac{2}{3} \log\bracket{\frac{1}{\sqrt{\cn}}+\frac{1}{\sqrt{1-\cn}}} - \log(2)
\geq \frac{2}{3} \log\bracket{1+\frac{1}{\sqrt{\cn}}} - \log(2) \\
&\geq \frac{2}{3} \log\bracket{1+\frac{1}{\sqrt{\cn}}} - \frac{\log(2)}{\log(3)}\log\bracket{1+\frac{1}{\sqrt{\cn}}}
\geq \frac{1}{30} \log\bracket{1+\frac{1}{\sqrt{\cn}}} \\
& \geq \frac{1}{60} \log\frac{1}{\cn}
\end{align*}

This holds with probability $\geq 1-\beta$, and we can convert it back to an expectation bound and have
\begin{align*}
\calL(\privT;D) - \calL(\theta^*;D) = \bigOmega{\log\frac{1}{\cn}}.
\end{align*}

\end{proof}

\subsection{Clipped Softmax Regression Does Not Correspond to a ``Natural'' Function}
\label{app:softMax}

Consider a $K$-class classification problem for $K \geq 3$.
Given a sample $(x,y)$ with $x\in\bfR^p$ and $y\in [K]$, the cross-entropy loss $\ell: \bfR^{p\times K} \times \bfR^p\times [K] \to \bfR$ is, for $\bt = [\theta^{(1)}, \dots, \theta^{(K)}]$,
\begin{align*}
\ell\bracket{\bt; (x,y)} = \sum_{k=1}^K \indi{y=k} \log \frac{\expo{\theta^{(k)}\cdot x}}{\sum_{k'=1}^K \expo{\theta^{(k')}\cdot x}}.
\end{align*}
We then have the gradient of $\ell$ as
\begin{align*}
\nabla_{\theta^{(k)}}\bracket{\bt; (x,y)} = \bracket{\frac{\expo{\theta^{(k)}\cdot x}}{\sum_{k'=1}^K \expo{\theta^{(k')}\cdot x}}-\indi{y=k}} \cdot x,
\end{align*}
and the clipped gradient as
$G(\bt) := \min\bracket{1, \frac{\cn}{\|\nabla_{\bt}\bracket{\bt; (x,y)}\|_2}} \cdot \nabla_{\bt}\bracket{\bt; (x,y)}$ for any $\bt \in \bfR^{p\times K}$
where $\nabla_{\bt}\bracket{\bt; (x,y)} = \left[\nabla_{\theta^{(1)}}\bracket{\bt; (x,y)},\dots, \nabla_{\theta^{(K)}}\bracket{\bt; (x,y)}\right]$.

\begin{proof}[Proof of Theorem~\ref{thm:clipped_softmax}]
Without loss of generality, let $y=1$. Let $x$ be any non-zero vector in $\bfR^p$. 

In the beginning, we state the formulas for the gradient, its norm and the clipped gradient.
Let $\ee{k}(\theta) = \expo{\theta^{(k')}\cdot x}$. (We omit $(\theta)$ when it is clear from the context.) 
Recall the gradient of the cross-entropy loss is
$
\nabla_{\theta^{(k)}}\bracket{\bt; (x,y)} = \bracket{\frac{\expo{\theta^{(k)}\cdot x}}{\sum_{k'=1}^K \expo{\theta^{(k')}\cdot x}}-\indi{y=k}} \cdot x
=\bracket{\frac{\ee{k}}{\sum_{k'=1}^K \ee{k'}}-\indi{y=k}} \cdot x
$,
so we have
\begin{align}\label{eqn:softmax gradient}
&\nabla_{\theta^{(1)}}\ell\bracket{\bt; (x,y)}
= -\frac{\sum_{k=2}^K \ee{k}}{\sum_{k'=1}^K \ee{k'}} \cdot x.
 \\
\text{For $k\geq 2$, } 
&\nabla_{\theta^{(k)}}\ell\bracket{\bt; (x,y)}
= \frac{\ee{k}}{\sum_{k'=1}^K \ee{k'}} \cdot x.
\end{align}
The norm of the gradient $\nabla_{\bt} \ell(\bt;(x,y))$ is thus
\begin{align}
\snorm{\nabla_{\bt}\bracket{\bt; (x,y)}}{2}
= \sqrt{\sum_{k=1}^K \snorm{\nabla_{\theta^{(k)}}\bracket{\bt; (x,y)}}{2}^2}
&= \|x\|_2 \cdot \frac{\sqrt{\power{\sum_{k=2}^K \ee{k}}{2} + \sum_{k=2}^K \power{\ee{k}}{2}}}{\sum_{k'=1}^K \ee{k'}},
\label{eqn:softmax proof grad norm}
\end{align}
which takes value in $\bracket{0, \sqrt{\frac{K}{K-1}}\|x\|_2}$. Recall that $\bT = \bbracket{\bt:\|\nabla_{\bt}\bracket{\bt; (x,y)}\|_2 > \cn}$.

Recall $G(\bt)$ is the clipped gradient. 
We also define, for $k \in [K]$, for $\bt \in \bfR^{p\times K}$,
\begin{align*}
G^{(k)}(\bt) := 
\min\bracket{1, \frac{\cn}{\|\nabla_{\bt}\bracket{\bt; (x,y)}\|_2}} \cdot \nabla_{\theta^{(k)}}\bracket{\bt; (x,y)},
\end{align*}
so $G(\bt) = \left[G^{(1)}(\bt),\dots,G^{(K)}(\bt)\right]$.

When $\bt \in \bT$, we have $G^{(k)}(\bt) = \cn \cdot  \frac{\nabla_{\theta^{(k)}}\ell\bracket{\bt; (x,y)}}{\snorm{\nabla_{\bt}\bracket{\bt; (x,y)}}{2}}$, and thus
\begin{align}
G^{(1)}(\bt)
&= -\frac{\cn}{\|x\|_2} \frac{\sum_{k=2}^K \ee{k}}{\sqrt{\power{\sum_{k=2}^K \ee{k}}{2} + \sum_{k=2}^K \power{\ee{k}}{2}}} \cdot x, \nonumber\\
\text{For $k\geq 2$, }
G^{(k)}(\bt)
&= \frac{\cn}{\|x\|_2} \frac{\ee{k}}{\sqrt{\power{\sum_{k=2}^K \ee{k}}{2} + \sum_{k=2}^K \power{\ee{k}}{2}}} \cdot x.
\label{eqn:softmax proof G}
\end{align}
Notice that for any $k\geq 2$, $\nabla_{\theta^{(1)}} G^{(k)}$ is zero as $G^{(k)}$ does not depend on $\theta^{(1)}$; however, $\nabla_{\theta^{(k')}} G^{(1)}$ may not be zero everywhere as $G^{(1)}$ does not depend on $\theta^{(k')}$ (we will prove this formally).

We will prove the theorem by contradiction.
Suppose there exists a function $f:\calC\to \bfR$ such that 1). $\bT\cap \interior{\calC}$ is a non-empty set, 2). $f$ is differentiable except for a set $\calC_N$ which is closed on $\calC$ and has zero measure, and 3) $G(\bt)$ is a subgradient of $f$.
We will show that on an open subset of $\bT\cap \calC$, $f$ is differentiable but the 2nd derivative is not symmetric, which contradicts the fact that any function with continuous second order partial derivative should have symmetry of 2nd derivative in the interior of its domain.

We use Euclidean topology throughout the proof. When not specified, we talk about Euclidean topology in the space $\bfR^{p\times K}$. We consider Lebesgue measure on $\bfR^{p\times K}$ throughout the proof.

\begin{enumerate}
\item 
First, we show $\bT$ %
is a non-empty open set in $\bfR^{p\times K}$.

Recall the formula for $\snorm{\nabla_{\bt}\bracket{\bt; (x,y)}}{2}$ in \eqref{eqn:softmax proof grad norm}, which is obviously a continuous function in $\bfR^{p\times K}$. Therefore, the preimage of open set $(\cn, \infty)$ through $\snorm{\nabla_{\bt}\bracket{\bt; (x,y)}}{2}$, which is exactly $\Theta$, is an open set in $\bfR^{p\times K}$. By assumption, $\bT$ is non-empty.

\item
Second, let $\bT_G = \bT\cap \bbracket{\bt : \forall k, k',\ \nabla_{\theta^{(k')}} G^{(k)}(\bt) = \nabla_{\theta^{(k)}} G^{(k')}(\bt)}$ be the ``good'' subset of $\bT$ where the 2nd derivative of $f$ is symmetric if $G$ is the derivative of $f$. 
We will show that $\bT_G$ is a closed set in $\bT$ and has Lebesgue measure $0$.

Recall $G^{(k)}(\bt)$ for $\bt \in \bT$ in~\eqref{eqn:softmax proof G}.
For any $k\geq 2$, notice that $G^{(k)}(\bt)$ does not depend on $\theta^{(1)}$, so $\nabla_{\theta^{(1)}}G^{(k)}(\bt) = \mathbf{0}$ for $k\geq 2$. 

Now we look at the derivatives of $G^{(1)}$. Let $D(\bt) = \power{\sum_{k=2}^K \ee{k}}{2} + \sum_{k=2}^K \power{\ee{k}}{2}$. 
For any $k'\geq 2$,
\begin{align*}
 \nabla_{\theta^{(k')}} G^{(1)}(\bt) 
=& -\frac{\cn\ee{k'}}{\|x\|_2 \power{D(\bt)}{3/2}} \sum_{k=2}^K \ee{k}\bracket{\ee{k}-\ee{k'}} xx^{\top}.
\end{align*}

As $\ee{k} > 0$ and $x \neq \mathbf{0}$, we have
\begin{align*}
\forall k'\geq 2,\ \nabla_{\theta^{(k')}} G^{(1)}(\bt) = \mathbf{0} 
& \Leftrightarrow
\ee{2} = \cdots = \ee{K} \\
& \Leftrightarrow
\ip{\theta^{(2)}}{x} = \cdots = \ip{\theta^{(K)}}{x}.
\end{align*}

It is also not hard to check that $\ee{2} = \cdots = \ee{K}$ is sufficient to guarantee $\nabla_{\theta^{(k')}} G^{(k)}(\bt) = \nabla_{\theta^{(k)}} G^{(k')}(\bt)$ for any $k,k'\geq 2$.

Therefore, we have $\bT_G = \bbracket{\bt\in \bT: \ip{\theta^{(2)}}{x} = \cdots = \ip{\theta^{(K)}}{x}}$.
We can define a function $a$ on $\bT$ with $a(\theta) = \sum_{k=3}^{K} \abs{\ip{\theta^{(2)}}{x} - \ip{\theta^{(k)}}{x}}$. Since $a$ is continuous on domain $\bT$, the preimage of the closed set $\{0\}$ through $a$, which is exactly $\bT_G$, is a closed set in $\bT$.
Also, $\bT_G$ is obviously a lower dimensional subspace of $\bfR^{p\times K}$ and thus has measure $0$.

\item
Third, let the ``bad'' set be $\bT_B = \bT\backslash \bT_G$. We will show $\bT_B\cap\interior{\calC}$ is a non-empty set and is open on $\interior{\calC}$.

As $\bT_G$ is closed in $\bT$, $\bT_B$, its complement, is an open set in $\bT$. %
As $\bT$ is open in $\bfR^{p\times K}$, $\bT_B$ is also open in $\bfR^{p\times K}$ (since $\bT_B$ is the intersection of two open subsets in $\bfR^{p\times K}$). So $\bT_B\cap\interior{\calC}$ is open on $\interior{\calC}$.

On the other hand, as $\bT$ and $\interior{\calC}$ are open, $\bT\cap\interior{\calC}$ is open. Additionally, by assumption, $\bT\cap\interior{\calC}$ is non-empty. So $\bT\cap\interior{\calC}$ has positive measure. 
Since $\bT_G$ has measure $0$, $\bT_B\cap\interior{\calC} = \bT\cap\interior{\calC}\backslash\bT_G$ has positive measure and is thus non-empty.

\item
Finally,
recall that $f$ is differentiable everywhere except for a closed set on $\calC_N$ with measure zero. Obviously, $\calC_N$ is also closed on $\interior{\calC}$. Then $f$ is differentiable on $\bT_B' := \bT_B\cap\interior{\calC}\backslash\calC_N$, which implies that $G$ is the gradient of $f$ on $\bT_B'$. Also, since $\forall k$, all partial derivatives of $G^{(k)}$ exists and is continuous, we know that $f$ has continuous 2nd derivatives on $\bT_B'$.

As $\bT_B\cap\interior{\calC}$ is open and $\calC_N$ is closed on $\interior{\calC}$, $\bT_B'$ is open on $\interior{\calC}$. Also, as $\calC_N$ has zero measure, $\bT_B'$ is non-empty.

By Schwarz's theorem, for any function that has continuous second order partial derivatives, it has symmetry of 2nd derivative in the interior of its domain.
So we are supposed to see $\nabla_{\theta^{(k)}}G^{(k')} = \nabla_{\theta^{(k')}}G^{(k)}$ for any pairs of $k$ and $k'$ on $\bT_B'$ (since $\bT_B'$ itself is non-empty and open in $\interior{\calC}$). However, this does not hold by definition of $\bT_B$. We therefore have a contradiction and such $f$ cannot exist.
\end{enumerate}
\end{proof}

\subsubsection{``Per-class'' Clipping Does Not Resolve the Problem}
\label{sec:per-class_clip}
\begin{thm}
Consider any sample $(x, y)$ with $x\in \bfR^p\backslash{\{\mathbf{0}\}}$, $y\in [K]$ (for $K \geq 3$) and any $\cn > 0$ such that $\bT = \bbracket{\bt: \|\nabla_{\theta^{(k)}} \ell(\bt;(x,y))\|_2 > \cn \text{ for some } k \in [K]}$ is non-empty.
Let $G(\theta)$ be the ``per-class'' clipped gradient of $\ell(\bt;(x,y))$. 
Consider any function $f:\calC \to \bfR$, $\calC\subseteq \bfR^{p\times K}$ such that $\bT \subseteq \calC$.
If $f$ is differentiable everywhere except for a set $\calC_N\subseteq \calC$
such that $\calC_N$ is a closed set on $\calC$ and has zero Lebesgue measure, then it is not possible for $\nabla_{\bt} f(\bt) = G(\bt)$ to hold for all $\bt \in \interior{\calC}\backslash\calC_N$.
\label{thm:clipped_softmax2}
\end{thm}

\begin{proof}
Recall the formula for the gradient in \eqref{eqn:softmax gradient} and the definition of $\ee{k}$, and we have
\begin{align*}
&\snorm{\nabla_{\theta^{(1)}} \ell(\bt; (x,y))}{2} = \frac{\sum_{k=2}^K \ee{k}}{\sum_{k'=1}^K \ee{k'}} \snorm{x}{2},\\
\text{For } k\geq 2,\ &\snorm{\nabla_{\theta^{(k)}} \ell(\bt; (x,y))}{2} = \frac{\ee{k}}{\sum_{k'=1}^K \ee{k'}} \snorm{x}{2}.
\end{align*}
Obviously, $\snorm{\nabla_{\theta^{(1)}} \ell(\bt; (x,y))}{2} = \sum_{k=2}^K \snorm{\nabla_{\theta^{(k)}} \ell(\bt; (x,y))}{2}$. So $\nabla_{\theta^{(k)}} \ell(\bt; (x,y))$ for $k \geq 2$ is clipped only when $\nabla_{\theta^{(1)}} \ell(\bt; (x,y))$ is clipped.
We consider when some of them are clipped, i.e. the set $\bT$. There are two cases.
\begin{enumerate}
\item If all of them are clipped, then $G^{(1)} = -\frac{\cn}{\snorm{x}{2}} x$ and $G^{(k)} = \frac{\cn}{\snorm{x}{2}} x$ for $k\geq 2$. $G$ is basically a constant and we have a valid gradient field.
\item If $\nabla_{\theta^{(1)}} \ell$ is clipped and some of $\nabla_{\theta^{(k)}} \ell$ for $k \geq 2$ is not clipped, then $G^{(1)} = -\frac{\cn}{\snorm{x}{2}} x$ and $G^{(k)} = \frac{\ee{k}}{\sum_{k'=1}^K \ee{k'}} x$. So we have $\nabla_{\theta^{(k)}} G^{(1)} = 0$ for any $k\geq 2$ and $\nabla_{\theta^{(1)}} G^{(k)} = -\frac{\ee{k} \ee{1}}{\power{\sum_{k'=1}^K \ee{k'}}{2}} x^\top x$ which is always nonzero. So we do not have a valid gradient field when this happens.

This is the set 
{\small
\begin{align*}
\bT_B = \cup_{k_0=2}^K \bT_{k_0} \\
\text{ where }
\bT_{k_0} &= \bbracket{\bt: \frac{\sum_{k=2}^K \ee{k}}{\sum_{k'=1}^K \ee{k'}} \snorm{x}{2} > \cn \text{ and }
\frac{\ee{k_0}}{\sum_{k'=1}^K \ee{k'}} \snorm{x}{2} \leq \cn} \\
&\supseteq \bbracket{\bt: \frac{\sum_{k=2}^K \ee{k}}{\sum_{k'=1}^K \ee{k'}} \snorm{x}{2} > \cn} \cap 
\bbracket{\bt:\frac{\ee{k_0}}{\sum_{k'=1}^K \ee{k'}} \snorm{x}{2} < \cn}
\end{align*}}
It is easy to see that $\bT_{k_0}$ is non-empty for any $k_0 \geq 2$. Also, $\bT_{k_0}$ is an open set as $\frac{\sum_{k=2}^K \ee{k}}{\sum_{k'=1}^K \ee{k'}}$ and $\frac{\ee{k_0}}{\sum_{k'=1}^K \ee{k'}}$ are continuous. So $\bT_B$ is an non-empty open set.

As $\bT \subseteq \calC$, we know $f$ is differentiable on $\bT \backslash \calC_N$ which is an non-empty open set.
Then $f$ cannot exists following the similar argument as in the proof of Theorem~\ref{thm:clipped_softmax}.
\end{enumerate}
\end{proof}

\section{Omitted Details for Experiments in Section~\ref{sec:expts}}
\label{app:expts}

Now, we provide details for our experimental setup that were omitted from the main body due to space constraints.

\mypar{Dataset} Following~\citep{jain2014near}, we use the Cod-RNA dataset~\citep{codrna} which contains $59,535$ training and $271,617$ validation (used as test set) samples, each with $8$ features and a binary label. 
To preprocess the data, we normalize each feature to range from $[-0.5, 0.5]$ using $f(x) = (x - \text{min}) / (\text{max} - \text{min}) - 0.5$, where $\text{min}$ and $\text{max}$ is the minimum and maximum value, respectively, of the feature in the training set. We then project each sample to the $\ell_2$-ball of radius $\sqrt{8/9}$, and append a constant feature of value $1/3$, such that each sample has $9$ dimensions and the $\ell_2$ norm is upper bounded by $1$.

To illustrate the effect of data dimensionality on accuracy, following~\citep{jain2014near} we append zero-valued features to the data samples, such that the accuracy of a non-private classifier does not change. We consider data dimensionality as $9$ (the original data) and 
$\cup_{i=1}^{4}\{1,2,5\} \times 10^{i}  \setminus \{10\}$.

\mypar{Model} We use logistic regression to solve the binary classification problem, i.e., for $d$-dimensional feature $x$ and label $y \in \{-1, 1\}$, the loss function is $\ell = \log(1 + e^{-y\theta^\top x})$ for model $\theta$. The gradient of the loss function is $\nabla \ell = {-yx}/(1 + e^{y\theta^\top x})$, which has bounded $\ell_2$ norm: $\|\nabla \ell\|_2 \leq \|x\|_2 \leq 1$.

\mypar{Algorithm}
To solve this problem privately, we consider DP-SGD with mini-batch gradient and Gaussian noise. Specifically, we use batch size $250$, and iterate over the training data for $10$ epochs.

As a comparison, we consider another algorithm with random noise drawn from $\Pr(Z=z) \propto e^{-\epsilon_0 \|z\|_2}$, i.e., the direction of the noise is uniform, and the magnitude of the noise follows from Gamma distribution with shape $d$ and scale $1/\epsilon_0$. 
Adding this random noise to the stochastic gradient (averaged over one batch) guarantees $\epsilon_0$-differential privacy. Thus, by privacy amplification, if one batch is formed by selecting each sample with probability $q < 0.5$, it guarantees $\epsilon = \log(1 + q(e^{\epsilon_0}-1))$-differential privacy~\citep{Ullman-lecturenotes}. By properties of R\'enyi differential privacy (RDP)~\citep{mironov2017renyi}, this converts to $(\alpha, \alpha\epsilon^2/2)$-RDP, and we can then use RDP composition to account privacy for multiple iterations.

\mypar{Privacy parameters and hyperparameter} 
We consider $\epsilon \approx 5.0$ and $\delta=10^{-5}$, which, under batch size $250$ and number of epochs $10$, can be achieved by setting the standard deviation of the Gaussian noise to be $\sigma=0.63$, and the scale parameter of the Gamma noise to be $0.57$.

The only hyperparameter in the experiment is the learning rate of DP-SGD. We search in the grid
$\cup_{i=-1}^1\{1, 2, 5\} \times 10^{i}$.
We run with every learning rate once, and pick the one with the highest averaged test accuracy over the last $5$ epochs. None of the chosen learning rate lies on the boundary of the grid.
Finally, we repeat each experiment for $5$ times with the chosen learning rate.

\fi
\end{document}